\documentclass[11pt]{article}
\def\withcolors{0}
\def\withnotes{0}
 
\usepackage[T1]{fontenc}
\usepackage[utf8]{inputenc}

\usepackage{lmodern}
\usepackage{xspace}

\usepackage{amsfonts,amsmath,amssymb, amsthm, mathtools}
\usepackage{thm-restate}
\usepackage{dsfont} 
\usepackage{algorithmicx,algpseudocode,algorithm}
\usepackage[usenames,dvipsnames,table]{xcolor}

\usepackage{mfirstuc}

\usepackage[backref,colorlinks,citecolor=blue,bookmarks=true,linktocpage]{hyperref}
\usepackage{aliascnt}
\usepackage[numbered]{bookmark}

\usepackage{fullpage}

\usepackage[shortlabels]{enumitem}
  \setitemize{noitemsep,topsep=3pt,parsep=2pt,partopsep=2pt}
  \setenumerate{itemsep=1pt,topsep=2pt,parsep=2pt,partopsep=2pt}
  \setdescription{itemsep=1pt}
  
\ifnum\withnotes=1
  \usepackage[colorinlistoftodos,textsize=scriptsize]{todonotes}
\fi

\usepackage{mleftright} 
\usepackage[mathscr]{eucal}
 
\makeatletter{}\makeatletter
\@ifundefined{theorem}{    \theoremstyle{plain}
    \newtheorem{theorem}{Theorem}[section]
  	\newaliascnt{coro}{theorem}
  	  
  	\aliascntresetthe{coro}
  	\newaliascnt{lem}{theorem}
  		\newtheorem{lemma}[lem]{Lemma}
  	\aliascntresetthe{lem}
  	\newaliascnt{clm}{theorem}
  		\newtheorem{claim}[clm]{Claim}
	\aliascntresetthe{clm}
	\newaliascnt{fact}{theorem}
 	 	\newtheorem{fact}[theorem]{Fact}
	\aliascntresetthe{fact}
  	
  \newaliascnt{prop}{theorem}
  		\newtheorem{proposition}[prop]{Proposition}
	\aliascntresetthe{prop}
	\newaliascnt{conj}{theorem}
  		
	\aliascntresetthe{conj}
 	 
  \theoremstyle{remark}   	\newtheorem{remark}[theorem]{Remark}

  \theoremstyle{definition}   	\newaliascnt{defn}{theorem}
 		 \newtheorem{definition}[defn]{Definition}
 	 \aliascntresetthe{defn}
 	   \theoremstyle{plain} }{}
\makeatother

\newenvironment{proofof}[1]{\noindent{\bf Proof of {#1}:~~}}{\hfill\(\QED\)}
\def\FullBox{\hbox{\vrule width 8pt height 8pt depth 0pt}}
\newcommand{\QED}{\;\;\;\FullBox}
\newenvironment{Proof}{\noindent{\bf Proof:~~}}{\hfill\QED}
\providecommand{\email}[1]{\href{mailto:#1}{\nolinkurl{#1}\xspace}}

\ifnum\withcolors=1
  \newcommand{\new}[1]{{\color{red} {#1}}}   \newcommand{\newer}[1]{{\color{blue} {#1}}}   \newcommand{\newest}[1]{{\color{orange} {#1}}}   \newcommand{\newerest}[1]{{\color{blue!10!black!40!green} {#1}}}   \newcommand{\acolor}[1]{{\color{cyan}#1}}   \newcommand{\ccolor}[1]{{\color{RubineRed}#1}}   \newcommand{\dcolor}[1]{{\color{purple}#1}}   \newcommand{\ecolor}[1]{{\color{ForestGreen}#1}}   \newcommand{\tcolor}[1]{{\color{orange}#1}}   \newcommand{\verynew}[1]{{\color{purple} #1}}
 \else
  \newcommand{\new}[1]{{{#1}}}
  \newcommand{\newer}[1]{{{#1}}}
  \newcommand{\newest}[1]{{{#1}}}
  \newcommand{\newerest}[1]{{{#1}}}
  \newcommand{\acolor}[1]{{#1}}
  \newcommand{\ccolor}[1]{{#1}}
  \newcommand{\dcolor}[1]{{#1}}
  \newcommand{\ecolor}[1]{{#1}}
  \newcommand{\tcolor}[1]{{#1}}
  \newcommand{\verynew}[1]{{#1}}
\fi

\ifnum\withnotes=1
\else
  \newcommand{\anote}[1]{}
  \newcommand{\cnote}[1]{}
  \newcommand{\dnote}[1]{}
  \newcommand{\enote}[1]{}
  \newcommand{\tnote}[1]{}
  \newcommand{\strikeout}[1]{}

\fi
\newcommand{\ignore}[1]{\leavevmode\unskip}      
\newcommand\numberthis{\addtocounter{equation}{1}\tag{\theequation}}
\newcommand{\eps}{\ensuremath{\epsilon}\xspace}
\newcommand{\Algo}{\ensuremath{\mathcal{A}}\xspace} \newcommand{\Tester}{\ensuremath{\mathcal{T}}\xspace}  \newcommand{\property}{\ensuremath{\mathcal{P}}\xspace}  \newcommand{\eqdef}{\stackrel{\rm def}{=}}

\newcommand{\accept}{\textsf{accept}\xspace}
\newcommand{\fail}{\textsf{fail}\xspace}
\newcommand{\reject}{\textsf{reject}\xspace}

\newcommand{\rhosubsetinf}{\shortexpect_{S \sim_\rho J}\left[ \infl{\phi_{\mI}(S)}\right]}
\newcommand{\ellsubsets}{\binom{[\ell]}{\ell-k}}

\newcommand{\bigO}[1]{{O\mleft( #1 \mright)}}

\newcommand{\bigOmega}[1]{{\Omega\mleft( #1 \mright)}}

\newcommand{\tildeO}[1]{\tilde{O}\mleft( #1 \mright)}
\newcommand{\tildeTheta}[1]{\operatorname{\tilde{\Theta}}\mleft( #1 \mright)}

\providecommand{\poly}{\operatorname*{poly}}

\newcommand{\infl}[2][f]{{\mathbf{Inf}_{#1}(#2)}}
\newcommand{\infldeg}[3][f]{{\mathbf{Inf}_{#1}^{#2}(#3)}}

\newcommand{\setOfSuchThat}[2]{ \left\{\; #1 \;\colon\; #2\; \right\} }
\newcommand{\indicSet}[1]{\mathds{1}_{#1}}
\newcommand{\indic}[1]{\indicSet{\left\{#1\right\}}}

\newcommand{\dist}[2]{\operatorname{dist}\mleft({#1, #2}\mright)}
\newcommand{\distiso}[2]{\operatorname{distiso}\mleft({#1, #2}\mright)}

\newcommand{\perm}[1][n]{\mathcal{S}_{#1}}

\newcommand\restr[2]{{  \left.\kern-\nulldelimiterspace   #1   \vphantom{\big|}   \right|_{#2}   }}

\newcommand{\proba}{\Pr}
\newcommand{\probaOf}[1]{\proba\!\left[\, #1\, \right]}
\newcommand{\probaCond}[2]{\proba\!\left[\, #1 \;\middle\vert\; #2\, \right]}
\newcommand{\probaDistrOf}[2]{\proba_{#1}\left[\, #2\, \right]}

\newcommand{\expect}[1]{\mathbf{E}\!\left[#1\right]}
\newcommand{\expectCond}[2]{\mathbf{E}\!\left[\, #1 \;\middle\vert\; #2\, \right]}
\newcommand{\shortexpect}{\mathbf{E}}

\newcommand{\exxp}[1]{\underset{#1}{\mathbf{E}}}
\newcommand{\uniform}{\ensuremath{\mathcal{U}}}

\newcommand{\binomial}[2]{\ensuremath{\operatorname{Bin}\!\left( #1, #2 \right)}}

\newcommand{\norm}[1]{\lVert#1{\rVert}}

\newcommand{\normtwo}[1]{{\norm{#1}}_2}
\newcommand{\norminf}[1]{{\norm{#1}}_\infty}
\newcommand{\abs}[1]{\left\lvert #1 \right\rvert}
\newcommand{\dabs}[1]{\lvert #1 \rvert}
\newcommand{\dotprod}[2]{ \left\langle #1,\xspace #2 \right\rangle } 			 			
 			
\newcommand{\clg}[1]{\left\lceil #1 \right\rceil}
\newcommand{\flr}[1]{\left\lfloor #1 \right\rfloor}

\newcommand{\R}{\ensuremath{\mathbb{R}}\xspace}

\newcommand{\N}{\ensuremath{\mathbb{N}}\xspace}

\newcommand{\pdfsamp}{dual\xspace}
\newcommand{\cdfsamp}{cumulative dual\xspace}
\newcommand{\Pdfsamp}{\expandafter\capitalisewords\expandafter{\pdfsamp}}
\newcommand{\Cdfsamp}{\expandafter\capitalisewords\expandafter{\cdfsamp}}

\newcommand{\lp}[1][1]{\ell_{#1}}

\newcommand{\D}{\ensuremath{D}}

\newcommand{\bool}{\{-1,1\}}
\newcommand{\pmbool}{\{\pm 1\}}
\newcommand{\junta}[1][k]{\mathcal{J}_{#1}}

\newcommand{\core}[1]{\ensuremath{\operatorname{\textsf{core}}_{#1}}}
\newcommand{\closejunta}[2][k]{\ensuremath{{#2}_{#1}}}
\newcommand{\kclose}{\ensuremath{k^{\ast}}}

\newcommand{\lovasz}[2][g]{\mathcal{L}_{#1}(#2)}
\newcommand{\oracle}{\mathcal{O}^{\pm}}
\newcommand*{\mL}{\mathcal{L}}

\newcommand{\oraclequery}[1]{\mathcal{O}_{#1}}

\newcommand{\mI}{\mathcal{I}}

\newcommand{\mS}{\mathcal{S}}

\makeatletter

\newcommand{\Rom}[1]{\expandafter\@slowromancap\romannumeral #1@}

\makeatother

\newcommand{\BE}{\begin{enumerate}} \newcommand{\EE}{\end{enumerate}}
\newcommand{\BI}{\begin{itemize}} \newcommand{\EI}{\end{itemize}}
\newcommand{\BT}{\begin{thm}} \newcommand{\ET}{\end{thm}}
 \newtheorem{thm}{Theorem}[section]   
\newtheorem{lem}[thm]{Lemma} \newcommand{\BL}{\begin{lemma}} \newcommand{\EL}{\end{lemma}}

\newtheorem{clm}[thm]{Claim}
\newcommand{\BCM}{\begin{clm}} \newcommand{\ECM}{\end{clm}}

\newtheorem{techcor}[thm]{Corollary}
\newcommand{\BCo}{\begin{techcor}} \newcommand{\ECo}{\end{techcor}}

\newtheorem{Conc}[thm]{Conclusion}
\newcommand{\BCONC}{\begin{Conc}} \newcommand{\ECONC}{\end{Conc}}

\newtheorem{Obs}[thm]{Observation}
\newcommand{\BOBS}{\begin{Obs}} \newcommand{\EOBS}{\end{Obs}}

\newtheorem{Exmp}[thm]{Example}
\newcommand{\BEXM}{\begin{Exmp}} \newcommand{\EXMP}{\end{Exmp}}

\newtheorem{cor} [thm] {Corollary}      \newcommand{\BC}{\begin{cor}} \newcommand{\EC}{\end{cor}}
\newtheorem{prop}[thm]{Proposition}     \newcommand{\BP}{\begin{prop}} \newcommand {\EP}{\end{prop}}

\newcommand{\subsetcoll}[2]{ \binom{#1}{#2} }

\def\FullBox{\hbox{\vrule width 8pt height 8pt depth 0pt}}

\newcommand{\BPF}{\begin{Proof}} \newcommand {\EPF}{\end{Proof}}
\newcommand{\BPFOF}{\begin{proofof}} \newcommand {\EPFOF}{\end{proofof}}

\newcommand{\BD}{\begin{definition}} \newcommand{\ED}{\end{definition}}

\makeatletter
\newcommand\ackname{Acknowledgements}
\if@titlepage
  \newenvironment{acknowledgements}{      \titlepage
      \null\vfil
      \@beginparpenalty\@lowpenalty
      \begin{center}        \bfseries \ackname
        \@endparpenalty\@M
      \end{center}}     {\par\vfil\null\endtitlepage}
\else
  
\fi
\makeatother

\def\authornameeb{Eric Blais}
\def\authoraffieb{University of Waterloo. Email: \email{eric.blais@uwaterloo.ca}. Research supported by NSERC Discovery grant.}
\def\authornamecc{Cl\'ement L. Canonne}
\def\authorafficc{Columbia University. Email: \email{ccanonne@cs.columbia.edu}. Research supported by NSF CCF-1115703 and NSF CCF-1319788.}
\def\authornamete{Talya Eden}
\def\authoraffite{School of EE, Tel Aviv University. Email: \email{talyaa01@gmail.com}. This research was partially supported by the Israel Science Foundation grant No. 671/13 and by a grant from the Blavatnik fund. }
\def\authornameal{Amit Levi}
\def\authoraffial{University of Waterloo. Email: \email{amit.levi@uwaterloo.ca}. Research supported by NSERC Discovery grant and the David R. Cheriton Graduate Scholarship.}
\def\authornamedr{Dana Ron}
\def\authoraffidr{School of EE, Tel Aviv University. Email: \email{danaron@post.tau.ac.il}. This research was partially supported by the Israel Science Foundation grant No. 671/13 and by a grant from the Blavatnik fund. }

\title{Tolerant Junta Testing and the Connection to Submodular Optimization and Function Isomorphism}
	\date{\today}

\author{
	\ecolor{\authornameeb}\thanks{\authoraffieb}
	\and \ccolor{\authornamecc}\thanks{\authorafficc}
	\and \tcolor{\authornamete}\thanks{\authoraffite}
	\and \acolor{\authornameal}\thanks{\authoraffial}
	\and \dcolor{\authornamedr}\thanks{\authoraffidr}
}

\begin{document}
	
	\maketitle

	\pagenumbering{gobble}	
	\begin{abstract}
		\makeatletter{}A function $f\colon \bool^n \to \bool$ is a  $k$-junta if it depends on at most $k$ of its variables. We consider the problem of \emph{tolerant} testing of $k$-juntas, where the testing algorithm must accept any function that is $\eps$-\emph{close} to some $k$-junta and reject any function that is $\eps'$-far from every $k'$-junta for some $\eps'= O(\eps)$ and $k' = O(k)$.

Our first result is an algorithm that solves this problem with query complexity polynomial in $k$ and $1/\eps$. This result is obtained via a new polynomial-time approximation algorithm for  \emph{submodular function minimization} (SFM) under large cardinality constraints, which holds even when only given an approximate oracle access to the function.

Our second result considers the case where $k'=k$. We show how to obtain a smooth tradeoff between the amount of tolerance and the query complexity in this setting. Specifically, we design an algorithm that given $\rho\in(0,1)$ accepts any function that is $\frac{\eps\rho}{16}$-close to some $k$-junta and rejects any function that is $\eps$-far from every $k$-junta. The query complexity of the algorithm is  $O\big( \frac{k\log k}{\eps\rho(1-\rho)^k} \big)$.

Finally, we show how to apply the second result to the problem of tolerant isomorphism testing between two unknown Boolean functions $f$ and $g$. We give an algorithm for this problem whose query complexity  only depends on the (unknown) smallest $k$ such that either $f$ or $g$ is \emph{close} to being a $k$-junta.
 
	\end{abstract}
	
	\ifnum\withnotes=1
	\clearpage
	\listoftodos
	\hfill
	\tableofcontents
	\clearpage
	\fi
	
	\clearpage

	\pagenumbering{arabic}
	\section{Introduction}\label{sec:intro}
	\makeatletter{}A function $f\colon \bool^n\to \bool$ is a \emph{$k$-junta} if it depends on at most $k$ of its variables.
Juntas are a central object of study in the analysis of Boolean functions, in particular
since they are good approximators for many classes of (more complex) Boolean functions. 
In the context of learning, the study of juntas was introduced by Blum et al.~\cite{blum1994relevant,blum1997selection} to model the problem of learning in the presence of irrelevant attributes. Since then, juntas have been extensively studied both in
computational learning theory (e.g.,~\cite{mossel2003learning,valiant2015finding})
and in applied machine learning (e.g.,~\cite{john2010elements}).

Juntas have also been studied within the framework of property testing. Here the task is to design a randomized algorithm that, given query access to a function $f$, accepts if $f$ is a $k$-junta and rejects if $f$ is $\eps$-far from every $k$-junta (i.e., $f$ must be modified in \tcolor{more than} an $\eps$-fraction of its values in order to be made a $k$-junta). The algorithm should succeed with high constant probability, and should
perform as few queries  as possible. The problem of testing $k$-juntas was first addressed by Fischer et.~\cite{FKRSS:04}.
They designed an algorithm that queries the function on a number of inputs  polynomial in $k$, and \emph{independent of $n$}. A series of subsequent works essentially settled the optimal query complexity for this  problem, establishing that\footnote{\tcolor{We use the notation $\tilde{\Theta}$  to hide polylogarithmic dependencies on the argument, i.e. for expressions of the form $\Theta(f \log^c f$ ) (for some absolute constant $c$).}} $\tildeTheta{k/\eps}$ queries are both necessary and sufficient~\cite{blais2008improved,blais2009testing,chockler2004lower,ServedioTW:15}.

The standard setting of property testing, however, is somewhat brittle, in that a testing algorithm is only guaranteed to accept functions that \emph{exactly} satisfy the property. But what if one wishes to accept functions that are \emph{close} to the desired property? To address this question, Parnas, Ron, and Rubinfeld introduced in~\cite{parnas2006tolerant} a natural generalization of property testing, where the algorithm is required to be \emph{tolerant}. Namely, a \emph{tolerant property testing algorithm} is required to accept any function that is \emph{close} to the property, and, as in the standard model, to reject any function that is far from the property.\footnote{Ideally, a tolerant testing algorithm should work for any given tolerance parameter $\eps'  < \eps$ (that is, accept functions that are $\eps'$-close to having the property), and have complexity that depends on $\eps-\eps'$. However, in some cases the relation between $\eps'$ and $\eps$ may be more restricted (e.g., $\eps' = \eps/c$ for a constant $c$). A closely related notion considered in~\cite{ parnas2006tolerant} is that of \emph{distance approximation} where the goal is to obtain an estimate of the distance that the object has to a property.}

As observed in~\cite{parnas2006tolerant},  any standard testing algorithm
whose queries are uniformly (but not necessarily independently) distributed,
is inherently tolerant to some extent. However, for many problems, strengthening the tolerance requires applying different methods and devising new algorithms (see e.g.,~\cite{GR:05,parnas2006tolerant,FN:07,ACCL,KS:09,MR:09,fattal2010approximating,CGR:13,BMR:16}).
Furthermore, there are some properties that have standard testers with sublinear query complexity, but for which any tolerant tester must perform a linear number of queries~\cite{FF:06, tell2016note}.

\sloppy
The problem of tolerant testing of juntas was previously considered by Diakonikolas et al.~\cite{diakonikolas2007testing}.
They applied the aforementioned observation from~\cite{parnas2006tolerant} and showed that one of the junta testers from~\cite{FKRSS:04} actually accepts functions that are
\dcolor{$\poly(\eps,1/k)$}-close
to $k$-juntas. Chakraborty et al.~\cite{CFGM:12}  observed that the analysis of the (standard) junta tester of Blais~\cite{blais2009testing}  implicitly implies  an  $\exp(k/\eps)$-query complexity tolerant tester which  accepts  functions that are $\eps/C$-close to some $k$-junta (for some constant $C > 1$) and rejects functions that are $\eps$-far from every $k$-junta.

\subsection{Our results}
In this work, we study the question of tolerant testing of juntas from two different angles, and obtain two algorithms with different (and incomparable) guarantees. Further, we show how to leverage one of these algorithms to get a tester for isomorphism between Boolean functions with ``instance-\tcolor{adaptive}'' (defined below) query complexity. The first of our results is a $\poly(k,1/\eps)$-query algorithm, which accepts functions that are close to $k$-juntas and rejects functions that are far from every $4k$-junta.

\begin{restatable}{theorem}{testingparameterized}
	\label{theo:tol:testing:juntas:relaxed}
	There exists an algorithm that, given query access to a function $f\colon\bool^n \to \bool$
	and parameters $k\geq 1$ and $\eps \in (0,1)$, satisfies the following.
	\begin{itemize}
		\item If $f$ is \verynew{$\eps/16$}-close to some $k$-junta, then the algorithm accepts
		with high constant probability.
		\item If $f$ is $\eps$-far from every $4k$-junta, then the algorithm rejects
		with high constant probability.
	\end{itemize}
		The query complexity of the algorithm is
	$
	\poly(k,\frac{1}{\eps})
	$ .
\end{restatable}

The algorithm referred to in the theorem can be seen as a relaxed version of a tolerant testing algorithm. Namely, the algorithm rejects functions that are $\eps$-far from every $4k$-junta rather than $\eps$-far from every $k$-junta. Similar relaxations have been considered both in the standard testing model~(e.g.,~\cite{parnas2002testing,kearns1998testing,kothari2014testing}) and in the tolerant testing model~\cite{parnas2002testing}.

We next study the question of tolerant testing without the above relaxation. That is, when the tester is required to reject functions that are $\eps$-far from being a $k$-junta. We obtain a \emph{smooth tradeoff} between the amount of tolerance and the query complexity. In particular, this tradeoff allows one to recover, as special cases, both the results of Fischer et al.~\cite{FKRSS:04} and (an improvement of) Chakraborty et al.~\cite{CFGM:12}.

\begin{restatable}{theorem}{testingtradeoff}
	\label{theo:tol:testing:juntas:tradeoff}
	There exists an algorithm that, given query access to a function $f\colon\bool^n \to \bool$
	and parameters $k\geq 1$, $\eps \in (0,1)$ and \new{$\rho\in(0,1)$}, satisfies the following.
	\begin{itemize}
		\item If $f$ is $\rho\eps/16$-close to some $k$-junta, then the algorithm accepts
		with high constant probability.
		\item If $f$ is $\eps$-far from every $k$-junta, then the algorithm rejects
		with high constant probability.
	\end{itemize}
		The query
	complexity of the algorithm is
	$
	\bigO{ \frac{k\log k}{\eps\rho(1-\rho)^k} }
	$.
\end{restatable}

Finally, we show how the above results can be applied to the problem of \emph{isomorphism testing}, which we recall next. Given query access to two unknown Boolean functions $f,g\colon\bool^n\to\bool$ and a parameter $\eps\in(0,1)$, one has to distinguish between {(i)}~$f$ is equal to $g$ up to some relabeling of the input variables; and {(ii)}~$\dist{f}{g\circ \pi} > \eps$ for every such relabeling $\pi$. The worst-case complexity of this task is known, with \tcolor{$\tilde{\Theta}\big({2^{\frac{n}{2}}}/\sqrt{\eps}\big)$} queries being necessary (up to the exact dependence on $\eps$) and sufficient~\cite{AB:10,ABCGM:13}.

However, is the exponential dependence on $n$ always necessary, or can we obtain better results for ``simple'' functions?
Ideally, we would like our testers to improve on this worst-case behavior, and instead have an \emph{instance-specific} query complexity, depending only on some intrinsic parameter of the functions $f,g$ to be tested. This is the direction we pursue here. Let $\kclose=\kclose(f,g,\gamma)$ be the smallest $k$ such that either $f$ or $g$ is $\gamma$-close to being a $k$-junta. We show that it is possible to achieve a query complexity only depending on this (unknown) parameter, namely of the form $\tilde{O}\big(2^{{\kclose(f,g,O(\eps))}/{2}}/\eps\big)$.\footnote{It is worth noting that this parameter can be much lower than the actual number of relevant variables for either functions; for instance, there exist functions depending on \emph{all} $n$ variables, yet that are $o(1)$-close to $O(1)$-juntas.} Moreover, our algorithm offers a much stronger guarantee: it allows \emph{tolerant} isomorphism testing.
\begin{restatable}[Tolerant isomorphism testing]{theorem}{isomorphism}
\label{theo:iso:testing:robust:kclose}
	There exists an algorithm that, given query access to two functions $f,g\colon\bool^n \to \bool$
	and parameter $\eps \in (0,1)$, satisfies the following, for some absolute constant $C \geq 1$.
	\begin{itemize}
		\item If $f$ and $g$ are $\frac{\eps}{C}$-close to isomorphic, then the algorithm accepts with high constant probability.
		\item If $f$ and $g$ are $\eps$-far from isomorphic, then the algorithm rejects with high constant probability.
	\end{itemize}
	The query	complexity of the algorithm is
	\tcolor{$\tilde{O}\big(2^{\frac{\kclose}{2}}/\eps\big)$} with high-probability (and \tcolor{$\tilde{O}\big(2^{\frac{n}{2}}/\eps\ \big)$} in the worst case), where $\kclose=\kclose(f,g,\frac{\eps}{C})$.
\end{restatable}
The above statement is rather technical, and requires careful parsing. In particular, the parameter $\kclose$ is crucially \emph{not} provided as input to the algorithm: instead, it is discovered adaptively by invoking the tolerant tester of~\autoref{theo:tol:testing:juntas:tradeoff}. This explains the high-probability bound on the query complexity: with some small probability, the algorithm may fail to retrieve the right value of $\kclose$ -- in which case it may use instead a larger value, possibly up to $n$.

\begin{remark}[On the running time of our algorithms.]
We note that, as in previous work on testing juntas, the query complexity depends only on $k$ and $1/\eps$
but the running time depends on $n$ (since even querying a single point in $\bool^n$ requires specifying $n$ bits).
\end{remark}

\subsection{Overview and techniques}\label{ssec:overview:techniques}
The proofs of Theorems~\ref{theo:tol:testing:juntas:relaxed} and~\ref{theo:tol:testing:juntas:tradeoff} both rely on the notion of the
\emph{influence} of a set of variables. Given a Boolean function $f\colon\bool^n\to\bool$ and a set $S\subseteq [n]$, the influence of the set $S$ (denoted $\infl[f]{S}$) is the probability that $f(x)\neq f(y)$ when $x$ and $y$ are selected uniformly subject to the constraint that for any $i\in \bar{S}$, $x_i=y_i$. The relation between the number of relevant variables and the influence of a set was utilized in previous works.
\dcolor{In what follows we} let $\junta[k]$ denote the set of all $k$-juntas.

Our starting point is similar to
\dcolor{the one in}~\cite{FKRSS:04,blais2009testing}.
We partition the $n$ variables into $\ell=O(k^2)$ parts, which allows us in a sense to remove the dependence on $n$. It is not hard to verify that if $f$ is close to $\junta[k]$, then there exist $k$ parts for which the following holds. If we denote by $T \subseteq [n]$ the union of variables in these $k$ parts, then the complement set $\bar{T}$ has small influence. On the other hand,  Blais~\cite{blais2009testing} showed that if a function is far from $\junta[k]$,  then a random partition into  a sufficiently large number of parts ensures the following with high constant probability. For every union $T$ of $k$ parts, the complement set $\bar{T}$ will have large influence. The above gives rise to a $(2^{(1+(o(1))k\log k}/\eps)$-query complexity  algorithm that distinguishes functions that are $\frac{1}{3}\eps$-close to
 $\junta[k]$  from functions that are $\eps$-far from $\junta[k]$. The algorithm considers all unions $T \subseteq [n]$ of $k$ parts, estimates the influence of $\bar{T}$, and accepts if there exists a set with sufficiently small estimated influence. In order to obtain an algorithm with better query complexity, we consider two relaxations.

In both relaxations we consider a fixed partition $\mI=\{I_1, \ldots, I_\ell\}$ of $\dcolor{[}n\dcolor{]}$ 
\dcolor{into $\ell = O(k^2)$ parts} and \dcolor{for $S\subseteq [\ell]$} we let $\phi_{\mI}(\dcolor{S}) \eqdef \dcolor{\bigcup}_{ i\in \dcolor{S}}I_i$.

\paragraph*{Parameterized tolerant testing through submodular minimization.} In order to describe the algorithm referred to in~\autoref{theo:tol:testing:juntas:relaxed}, it will be useful to introduce the following function. For a Boolean function $f$ and a partition $\mI \dcolor{= \{I_1,\dots,I_{\verynew{\ell}} \}}$, we let $h\colon 2^{[\verynew{\ell}]} \to [0,1]$ be defined as $h(J)\eqdef \infl[f]{\phi_{\mI}(J)}$. The starting point of our approach is the observation that the exhaustive search algorithm described previously can be seen as performing a brute-force minimization of $h$, under a cardinality constraint. Indeed, it effectively goes over all sets $J\subseteq [\verynew{\ell}]$ of size $\verynew{\ell}-k$, estimates $h(J)$, and accepts if there exists a set $J$ for which the estimated value is sufficiently small. With this view, it is natural to ask whether this minimization can be performed more efficiently, by exploiting the
\dcolor{fact that} \acolor{by the diminishing marginal property of the influence}, $h$ is submodular. That is, for every two sets $J_1 \subseteq J_2$ and variable $i \notin J_1$, it holds that $h(J_1\cup\{i\}) - h(J_1) \geq h(J_2\cup\{i\}) - h(J_2)$. While it is possible to find the minimum value of a submodular function in polynomial time \new{ given query access} \dcolor{to the function}, if a cardinality constraint is introduced, then even finding an approximate minimum is hard~\cite{SF:11}. In light of the hardness of the problem, we design an algorithm for the following \emph{bi-criteria relaxation}. Given oracle access to a non-negative submodular function $h\colon 2^{\new{[}\verynew{\ell}\new{]}} \to \R$ and input parameters $\eps \in (0,1)$ and $k \in \N$, the algorithm distinguishes between the following \dcolor{two} cases:
\BI
  \item There exists a set $J$ such that $\abs{J}\geq \verynew{\ell}-k$ and $h(J) \leq \eps$;
  \item For every set $J$ such that $\abs{J} \geq \verynew{\ell}-2k$, $h(J)> 2\eps$.
\EI

Moreover, the algorithm can be adapted to the case where it is only granted access to an approximate oracle for $h$ (for a precise statement, see~\autoref{thm:ASMC:correct}). This is critical in our setting, since $h(J)=\infl[f]{\phi_{\mI}(J)}$, and we can only estimate the influence of sets of variables.

\paragraph{Subset influence and recycling queries.}

The key idea behind our second approach is the following.  The exhaustive search algorithm estimates the influence of 
\dcolor{the set of variables $\phi_{\mI}(J)$ for every set of indices $J \subset [\ell]$ such that $|J|=\ell-k$}
by performing pairs of queries specifically designed for $J$. Namely, it queries the value of the function on pairs of points
\dcolor{in $\bool^n$}  that agree on the set $\bar{J}$. 
If it \dcolor{were} possible to use the same queries for estimating the influence of \dcolor{$\phi_{\mI}(J)$ for} different  \dcolor{choices of $J$},
then we could reduce the query complexity. We show that this can be done if we consider the {\em $\rho$-biased subset influence\/} of a set $J \subset [\ell]$, defined next.

 Given a partition $\mI=\{I_1, \ldots, I_{\ell}\}$,  a parameter \new{$\rho\in(0,1)$}, and a set $J\subset [\ell]$, a random $\rho$-biased subset $S \sim_\rho J$ is a subset of $J$ resulting from taking every index in $J$ to $S$ with probability $\rho$. The expected influence of a random $\rho$-biased subset of $J$, referred to as the $\rho$-subset influence of $J$, is $\shortexpect_{S \sim_\rho J}\left[ \infl{\phi_{\mI}(S)}\right]$. We prove that for every set $J \subseteq [\ell]$, its $\rho$-subset influence is in $[\frac{\rho}{3}\infl{\phi_{\mI}(J)}, \infl{\phi_{\mI}(J)}]$. \newerest{A crucial element in our proof is a combinatorial result due to Baranyai~\cite{Baranyai:75} on factorization of regular hypergraphs. With this fact in hand, we then} present an algorithm that allows to simultaneously estimate the $\rho$-subset influence of all sets $J\subset [\ell]$ of size $\ell -k$. \tcolor{The query complexity of the algorithm is $\bigO{\frac{k\log \new{k}}{\eps\rho(1-\rho)^k}}$.}

\paragraph{Application to isomorphism testing: tolerant testing and noisy samplers.}

The structure of our tolerant isomorphism testing algorithm is quite intuitive, and consists of two phases. In the first phase, we run a linear search on $k$, repeatedly invoking our tolerant \dcolor{junta} tester to discover the smallest value $k$ satisfying $\min(\dist{f}{\junta[k]},\dist{g}{\junta[k]}) \leq \eps/C$. We note that a similar approach using a tester whose tolerance is only $\poly(\eps/k)$ might return a much larger value of $k$, since as $k$ increases, the allowed tolerance decreases. In the second phase, we use this value of $k$ to tolerantly test isomorphism between $f$ and $g$. This phase, however, is not as straightforward as it seems: indeed, to achieve the desired query complexity, we would like to test isomorphism~--~for which we have known algorithms~--~between $\closejunta[k]{f}$ and $\closejunta[k]{g}$, that is, the $k$-juntas closest to $f$ and $g$ respectively.

Yet here we face two issues: (i) we do not have query access to $\closejunta[k]{f}$ and $\closejunta[k]{g}$;
 (ii) even in the completeness case $f_k$ and $g_k$ \emph{need not actually be isomorphic.} Indeed, $f$ and $g$ are only promised to be {close} to $k$-juntas, and close to isomorphic. Hence, the corresponding juntas are only guaranteed to be \emph{close} to isomorphic.

Addressing item~{(ii)} relies on adapting the algorithm of~\cite{ABCGM:13}, along with a careful and technical analysis of the distribution of the points it queries. (This analysis is also the key to providing the tolerance guarantees of our isomorphism tester.) We address item~{(i)} as follows. Our algorithm builds on the ideas of Chakraborty et al.~\cite{CGM:11}, namely on their notion of a ``noisy sampler''. A noisy sampler is given query access to a function that is promised to be close to some $k$-junta and provides (almost) uniformly distributed samples labeled (approximately) according to this $k$-junta. While the~\cite{CGM:11} noisy sampler works for functions that are $\poly(\eps/k)$-close to $\junta$, we need a noisy sampler that works for functions that are only $\frac{\eps}{C}$-close to $\junta$. To this end, we replace the weakly tolerant testing algorithm of~\cite{blais2009testing} used in the noisy sampler of~\cite{CGM:11} with our tolerant testing algorithm. The query complexity of the resulting noisy sampler is indeed much higher than that of~\cite{CGM:11}. However, this does not increase the overall query complexity of our tolerant isomorphism testing algorithm, as stated in \autoref{theo:iso:testing:robust:kclose}.

\subsection{Organization of the paper}

After introducing the necessary notations and definitions in~\autoref{sec:prelim}, we describe in~\autoref{sec:basics} the common starting point of our algorithms --~the reduction from $n$ variables to $O(k^2)$ parts. \autoref{sec:ASMC} then contains the details of the submodular minimization under cardinality constraint underlying~\autoref{theo:tol:testing:juntas:relaxed}, which is then implemented in~\autoref{sec:ASMM} with an approximate submodular minimization primitive. We then turn in~\autoref{sec:tradeoff} to the proof of~\autoref{theo:tol:testing:juntas:tradeoff}, before describing in~\autoref{sec:isomorphism} how to leverage it to obtain our instance-adaptive tolerant isomorphism testing result.  
	
	 \section{Preliminaries}\label{sec:prelim}
	 \makeatletter{}\subsection{Property testing, tolerance, and juntas}
A \emph{property} $\property$ of Boolean functions is a subset of all these functions, and we say \tcolor{that} a function $f$ \emph{has the property $\property$} if $f\in \property$. The distance between two functions $f,g\colon\bool^n\to \bool$ is defined as their (normalized) Hamming distance $\dist{f}{g}\eqdef\probaDistrOf{x}{f(x)\neq g(x)}$, where $x$ is drawn uniformly at random. Accordingly, for a function $f$ and a property $\property$ we define the distance from $f$ to $\property$ as $\dist{f}{\property}\eqdef\min_{g\in\property}\dist{f}{g}$. Given $\eps \geq 0$ and a property $\property$, we will say \tcolor{that} a function $f$ is $\eps$-\emph{far} from $\property$ (resp. $\eps$-\emph{close} to $\property$) if $\dist{f}{\property} > \eps$ (resp. $\dist{f}{\property} \leq \eps$).\medskip

\noindent We can now give a formal definition of a property testing algorithm.
 \begin{definition}
  A \emph{testing algorithm for a property $\property$} is a probabilistic algorithm that gets an input parameter $\eps\in (0,1)$ and oracle access to a function $f\colon\bool^n \to \bool$. The algorithm should output a binary verdict that satisfies the following two conditions.
 	\begin{itemize}
 	\item If $f\in \property$, then the algorithm accepts $f$ with probability at least $2/3$.
 	\item If $\dist{f}{\property} > \eps$, then the algorithm rejects $f$ with probability at least $2/3$.
 	\end{itemize}
 \end{definition}

Next, we define the notion of \tcolor{a} \textit{tolerant} testing algorithm, a testing algorithm that is also required to accept functions merely close to the property:
\begin{definition}
	A \emph{tolerant testing algorithm for a property $\property$} is a probabilistic algorithm that gets two input parameters $\eps_1,\eps_2\in(0,1)$ such that $\eps_1< \eps_2$, and oracle access to a function $f\colon\bool^n \to \bool$.
	The algorithm should output a binary verdict that satisfies the following two conditions.
 	\begin{itemize}
 	\item If $\dist{f}{\property} \leq \eps_1$, then the algorithm accepts $f$ with probability at least $2/3$.
 	\item If $\dist{f}{\property} > \eps_2$, then the algorithm rejects $f$ with probability at least $2/3$.
 	\end{itemize}
In some cases the algorithm is only given one parameter, $\eps_2$, setting 
$\eps_1 = r(\eps_2)$ for some prespecified function \tcolor{$r\colon(0,1) \to (0,1)$}.
\end{definition}

We also consider a relaxation of the definition of tolerant testing to the following \textit{tolerant testing of parameterized properties} .
\begin{definition}[Tolerant Testing of Parameterized Properties]
	Let $\property=(\property_s)_{s\in\N}$ be a non-decreasing family of properties parameterized by $s\in \N$, i.e. such that $\property_s\subseteq \property_t$ whenever $s\leq t$; and $\sigma\colon \N\to \N$ be a non-decreasing mapping. A \emph{$\sigma$-tolerant testing algorithm for $\property$} is a probabilistic algorithm that gets three input parameters $s\in\N$ and $\eps_1,\eps_2\in[0,1]$ such that $\eps_1< \eps_2$, as well as oracle access to a function $f\colon\bool^n \to \bool$.
	The algorithm should output a binary verdict that satisfies the following two conditions.
 	\begin{itemize}
 	\item If $\dist{f}{\property_s} \leq \eps_1$ then the algorithm accepts $f$ with probability at least $2/3$.
 	\item If $\dist{f}{\property_{\sigma(s)}} > \eps_2$, then the algorithm rejects $f$ with probability at least $2/3$.
 	\end{itemize}
 \dcolor{Here too $\eps_1$ may be a prespecified function  of $\eps_2$.}
\end{definition}

The main focus of this work will be the property of being a \emph{junta}, that is, a Boolean function that only depends on a (small) subset of its variables:

 \begin{definition}[Juntas]
 	A Boolean function $f\colon \pmbool^{n}\to \pmbool$ is a \emph{$k$-junta} if there exists a set $\tcolor{T}\subseteq [n]$ of size at most $k$, such that $f(x)=f(y)$ for every two assignments $x,y\in\pmbool^n$ that satisfy $x_i=y_i$ for every $i\in \tcolor{T}$. We let $\junta[k]$ denote the set of all $k$-juntas (over $n$ variables).
 \end{definition}

\paragraph*{Notations.} Hereafter, we denote by $\log$ the binary logarithm, by $[n]$ the set of integers $\{1,\dots,n\}$, and by $\perm[n]$ for the set of permutations of $[n]$. 
Given two disjoint sets $S,T\subseteq [n]$ and two partial assignments $x\in \pmbool^{\abs{S}}$ and $y\in\pmbool^{\abs{T}}$, we let $x\sqcup y \in \pmbool^{\abs{S\cup T}}$ be the partial assignment whose $i$-th coordinate is $x_i$ if $i\in S$ and $y_i$ if $i\in T$. Given a Boolean function $f\colon\bool^n\to\bool$ we write $\oraclequery{f}$ for an oracle providing query access to $f$. \verynew{For a set $S$, We denote by $\subsetcoll{S}{r}$ the set of all subsets of $S$ of size $r$. Given a partition of $\mI=\{I_1, \ldots, I_\ell \}$ of $[n]$ and a set $J \subseteq [\ell]$, we denote by $\phi_{\mI}(J)$ the union $\bigcup_{i \in J}I_i$. }

\medskip
A key notion in this work is the notion of \emph{influence} of a set, which generalizes the standard notion of influence of a variable:
\begin{definition}[Set-influence]\label{def:set:influence}
	For a Boolean function $f\colon\bool^n\to\bool$, the \emph{set-influence} of a set $S\subseteq[n]$ is defined as
	\[
	\infl{S} = 2\probaOf{f(x\sqcup u)\neq f(x\sqcup v)}\;,
	\]
	where $x\sim{\bool^{[n]\setminus S}}$, and $u,v \sim{\bool^{S}}$.
\end{definition}

	\section{From $n$ variables to $O(k^2)$ parts}\label{sec:basics}
	\makeatletter{}In this section we build on techniques from~\cite{FKRSS:04,blais2009testing} and describe how to reduce the problem of testing closeness to a $k$-junta to testing closeness to a \emph{$k$-part} junta (defined below). The advantage of doing so is that while the former question  concerns functions on $n$ variables, the latter \emph{does no longer involve $n$ as a parameter}: only $k$ and $\eps$ now have a role to play. We start with
\dcolor{a useful definition of \emph{$k$-part juntas}, and two lemmas regarding their properties with respect to
random partitions of the domain.}

\begin{definition}[Partition juntas {\cite[Definition 5.3]{Blais:PhD:12}, extended}] \label{def:partition_juntas}
Let $\mI$ be a partition of $[n]$ into $\ell$ parts, and $k\geq 1$. The function $f\colon\bool^n\to\bool$ is a \emph{$k$-part junta with respect to $\mI$} if the relevant coordinates in $f$ are all contained in at most $k$ parts of $\mI$. Moreover,
\begin{enumerate}[(i)]
	\item $f$ is said to \emph{$\eps$-approximate being a $k$-part junta with respect to $\mI$} if there exists a set $J \tcolor{\in \subsetcoll{[\ell]}{\ell - k}}$  satisfying $\infl[f]{\tcolor{\phi_{\mI}(J)}} \leq 2\eps$.
	\item Conversely, $f$ is said to \emph{$\eps$-violate being a $k$-part junta with respect to $\mI$} if for every set $J \tcolor{\in \subsetcoll{[\ell]}{\ell-k} }$, $\infl[f]{\tcolor{\phi_{\mI}(J)}} > 2\eps$.
\end{enumerate}
\end{definition}

\begin{lemma}[{\cite[Lemma 5.4]{Blais:PhD:12}}]\label{lemma:random:partition:soundness}
For  $f\colon\bool^n\to\bool$ and $k \geq 1$, let  $\alpha\eqdef \dist{f}{\junta[k]}$. Also, let $\mI$ be a random partition of $[n]$ with $\ell \eqdef 24k^2$ parts obtained by uniformly and independently assigning each coordinate to a part. With probability at least $5/6$ over the choice of the partition \dcolor{$\mI$, the function} $f$ $\frac{\alpha}{2}$-violates being a
$k$-part junta with respect to $\mI$.
\end{lemma}

\begin{lemma}\label{lemma:any:partition:completeness}
For  $f\colon\bool^n\to\bool$ and $k \geq 1$, let $\alpha\eqdef \dist{f}{\junta[k]}$ and let $\mI$ be any partition of $[n]$ into $\ell \geq k$ parts. Then $f$ $2\alpha$-approximates being a $k$-part junta with respect to $\mI$.
\end{lemma}
\BPF
Let $g\in\junta[k]$ be such that $\dist{f}{g} = \dist{f}{\junta} = \alpha$. Let $I_{i_1},\dots, I_{i_r}$ be the $r\leq k$ parts of $\mI$ containing the relevant variables of $g$. Then, for any set $J\subset[\ell]$ of size $\tcolor{\ell-} k$ such that $\{i_1,\dots,i_r\}\subseteq \tcolor{\bar{J}}$, we have that 
when drawing $x\sim {\bool^{\verynew{\phi_{\mI}(\bar J)}}}$, and $u,v \sim{\bool^{\verynew{\phi_{\mI}({J})}}}$ the following holds.
\begin{align*}
  \infl{\verynew{\phi_{\mI}({J})}} &= 2\probaOf{f(x\sqcup u)\neq f(x\sqcup v)}
  \leq 2\probaOf{f(x\sqcup u)\neq g(x\sqcup u) \text{ or } f(x\sqcup v)\neq g(x\sqcup v)} \\
  &\leq 2\left(\probaOf{f(x\sqcup u)\neq g(x\sqcup u) } + \probaOf{f(x\sqcup v)\neq g(x\sqcup v)} \right)
  \leq 2\left( \alpha + \alpha \right)
  = 4\alpha \;,
\end{align*}
where the first inequality follows from observing that (as $g$ does not depend on variables in $\verynew{\phi_{\mI}}({J})$) one can only have $f(x\sqcup u)\neq f(x\sqcup v)$ if $f$ disagrees with $g$ on at least one of the two points; and the third inequality holds since both $x\sqcup u$ and $x\sqcup v$ are uniformly distributed.
\EPF\medskip

The above two lemmas suggest the following approach
\dcolor{for distinguishing between} functions that are $\eps'$-close to some $k$-junta and functions that are $\eps$-far from every $k'$-junta. Suppose we select a random partition of $[n]$ into $O(k^2)$ parts. Then, with high probability over the choice of the partition, it is sufficient to distinguish between functions that $2\eps'$-approximate being a $k$-junta and functions that $\eps/2$-violate being a $k'$-part junta.
 Specifically, we get the proposition below, which we apply throughout this work:

\begin{proposition}[Reduction to part juntas]\label{prop:tol:testing:juntas:reduction:n:m}
Let $\Tester$ be an algorithm that is given query access to a function $f:\bool^n \to \bool$, a partition $\mI=\{I_1, \ldots, I_{\new{\ell}}\}$ of $[n]$  into $\new{\ell}$ parts, and parameters  $k \in \N$ and $\eps\in (0,1)$. Suppose that $\Tester$ performs $q(k,\eps,\new{\ell})$ queries to $f$ and satisfies the following guarantees, for a pair of functions $\new{r} \colon(0,1) \times \N \to (0,1)$ and $\new{r'}\colon\N \to \N$.
\begin{itemize}
\item If $f$ $\eps'$-approximates being a $k$-part junta with respect to $\mI$ and $\eps'\leq \new{r}(\eps,k)$, then $\Tester$ returns \accept with probability at least $5/6$;
\item If $f$ $\eps$-violates being a $k'$-part junta with respect to $\mI$ and $k' \geq \new{r'}(k)$, then $\Tester$ returns \reject with probability at least $5/6$.
\end{itemize}
Then there exists an algorithm $\Tester'$, that given query access to $f$ and parameters  $k \in \N$ and $\eps \in (0,1)$, satisfies the following.
\begin{itemize}
\item If $\dist{f}{\junta[k]} \leq \frac{\eps'}{2}$ and $\eps'\leq \new{r}(\eps, k)$, then $ \Tester'$ outputs \accept with probability at least $2/3$;
\item If $\dist{f}{\junta[k']} > 2\eps$ and $k' \geq \new{r'}(k)$, then $\Tester'$ outputs \reject with probability at least $2/3$.
\end{itemize}
Moreover, the algorithm $\Tester'$ has query complexity $q(k,\eps,\acolor{\ell})$.
\end{proposition}

\BPFOF{\autoref{prop:tol:testing:juntas:reduction:n:m}} The algorithm $\Tester'$ first obtains a random partition $\mI$ of $[n]$ into $\new{\ell}\eqdef 24\dcolor{(}k'\dcolor{)}^2$ parts by uniformly and independently assigning each coordinate to a part. $\Tester'$ then invokes $\Tester$ with parameters $\eps,k,\new{\ell}$ and the partition $\mI$. By \autoref{lemma:random:partition:soundness}  and the choice of $\new{\ell}$, with probability at least $5/6$ the partition $\mI$ is \emph{good} in the following sense. For $\alpha=\dist{f}{\junta[k']}$, it holds that $f$ $\frac{\alpha}{2}$-violates being a $k'$-junta with respect to $\mI$. Conditioned on $\mI$ being good, and by~\autoref{lemma:any:partition:completeness}, we are guaranteed that the following holds.
	\begin{enumerate}[(i)]
		\item If $\dist{f}{\junta[k]} \leq \frac{\eps'}{2}$, then $f$ $\eps'$-approximates being a $k$-part junta with respect to $\mI$;
		\item If $\dist{f}{\junta[k']} > 2\eps$, then $f$ $\eps$-violates being a $k'$-part junta with respect to $\mI$.
	\end{enumerate}
	Therefore, $\Tester$ will answer as specified by the proposition with probability at least $5/6$, making $q(\eps,k,\new{\ell})$ queries. Overall, \dcolor{by a union bound,} $\Tester'$ is successful with probability at least $2/3$.
 \EPFOF

\paragraph{} As an illustration of the above technique, and a warmup towards the (more involved) algorithms of the next sections, we show how to obtain an algorithm $\Tester'$ as specified in~\autoref{prop:tol:testing:juntas:reduction:n:m} with query complexity $2^{(1+o(1))k \log k}/\eps$.
Given a partition $\mI$ of $[n]$ into $\new{\ell}$ parts, $\Tester$ considers all $\binom{\new{\ell}}{\ell- k}$ sets of variables that result from taking the union of $k$ parts. For each such set $T$, it obtains an estimate $\widetilde{\mathbf{Inf}}_f({T})$ of the influence of ${T}$, by performing \tcolor{$O\left(\frac{\ell \log \ell}{\eps}\right)$} queries to $f$. $\Tester$ accepts if for at least one of the sets $T$, $\widetilde{\mathbf{Inf}}_f({T})$ is at most $\frac{3}{2}\eps$. Performing \tcolor{$O(\frac{\ell \log \new{\ell}}{\eps})$} queries to the oracle for each set, ensures that the following holds with high constant probability. For every set $T$ such that $\infl[f]{T} \leq \frac{4}{3}\eps$, $\widetilde{\mathbf{Inf}}_f(\new{T}) \leq \frac{3}{2}\eps$ and for every set $T$ such that ${\mathbf{Inf}}_f(T) > 2\eps$, $\widetilde{\mathbf{Inf}}_f(T) > \frac{3}{2}\eps$. Hence, the algorithm $\Tester$ fulfills the requirements stated in~\autoref{prop:tol:testing:juntas:reduction:n:m} \new{(for $\new{r}(\eps,k)=\frac{2}{3}\eps$ and $\new{r}'(k)=k$)}, and it follows that:
 \BI
	 \item  If $f$ is $\frac{1}{3}\eps$-close to some  $k$-junta then $\Tester'$ accepts with probability at least $2/3$.
	 \item  If $f$ is $\eps$-far from every $k$-junta then $\Tester'$ rejects with probability at least $2/3$.
 \EI
 Since $\new{\ell}=24k^2$, the query complexity of the algorithm is $\binom{\new{\ell}}{\ell -k} \cdot O(\frac{\tcolor{\ell }\log \ell}{\eps}) = 2^{(1+o(1))k\log k}/\eps$.

	\section{Approximate submodular minimization under \dcolor{a} cardinality constraint}\label{sec:ASMC}
	\makeatletter{}In this section we show how a certain bi-criteria approximate version of submodular minimization with a cardinality constraint can be reduced to approximate submodular minimization with no cardinality constraint. This reduction holds even when given \emph{approximate oracle access} to the submodular function, and is meaningful when the cardinality constraint is sufficiently large. Precise details follow.

\BD[Approximate oracle]\label{def:approx-oracle}
Let $h\colon 2^{[\verynew{\ell}]}\to \R$ be a function. An \emph{approximate oracle} for $h$, denoted $\oracle_h$, is a randomized algorithm \dcolor{that}, for any input $J \subseteq [\verynew{\ell}]$ and parameters $\tau,\delta\in (0,1)$, returns a value $\tilde{h}(J)$ such that $|\tilde{h}(J) - h(J)| \leq \tau$ with probability at least $1-\delta$.
\ED

\BD[Approximate submodular minimization algorithm] \label{def:ASMM} Let $h\colon 2^{[\verynew{\ell}]} \to \R$ be a non-negative submodular function and let $\oracle_h$ denote an approximate oracle for $h$.
An \emph{approximate submodular function minimization} algorithm (ASFM) is an algorithm that, when given access to $\oracle_h$ and called with input parameters $\xi$ and $\delta$, returns a value $\nu$ such that $|\nu -\min_{\verynew{J\subseteq[\ell]}} \{h(\verynew{J})\} |\leq  \xi$ with probability at least $1-\delta$.
\ED
In~\autoref{coro:submodular:noisy:minimization} in~\autoref{sec:ASMM} we establish the existence of such an \dcolor{ASFM} algorithm. The running time of the algorithm is polynomial in $\verynew{\ell}$, logarithmic in the maximal value of the function and linear in the running time of the approximate oracle. We next present an algorithm for approximate submodular minimization under \dcolor{a} cardinality constraint.

\begin{algorithm}[H]
\caption{\dcolor{Approximate} Submodular Minimization under \dcolor{a Cardinality} Constraint($\oracle_h, \eps, \delta, \xi, k$)\label{alg:ASMC}}\begin{algorithmic}[1]
	
	\State Let $h'(J)=h(J)-\frac{\eps}{k}\abs{J}$ so that for every $\tau',\delta'$ $\oracle_{h'}(J,\tau',\delta')=\oracle_h(J, \tau',\delta')-\frac{\eps}{k}\abs{J}$ \label{step:def_h}.
	
	\State Let $\nu$ be the returned value from invoking an ASFM algorithm with access to $\oracle_{h'}$ and  parameters  $\xi$ and $\delta$. \label{step:def:nu}
	\State Accept if and only if $\nu\leq (1-\frac{\verynew{\ell}-k}{k})\cdot \eps + \xi$.
\end{algorithmic}
\end{algorithm}

\BT\label{thm:ASMC:correct}
For a submodular function $h$,~\autoref{alg:ASMC} satisfies the following conditions:
\BE
\item If there exists a set \tcolor{$J \subseteq [\ell]$} such that $\abs{J} \geq \ell-k$ and and $h(J)\leq \eps$, then the algorithm accepts with probability at least $1-\delta$.
\item If for every set \tcolor{$J \subseteq [\ell]$} such that $\abs{J} \geq \ell-2(1+\frac{\xi}{\eps})k$, we have  $h(J) > 2(\eps+\xi)$, then the algorithm rejects with probability at least $1-\delta$.
\EE
Moreover, the second item can be strengthened so that it holds for functions $h$ \dcolor{that} satisfy the following: (i) for every set \tcolor{$J\subseteq [\ell]$} such that  $\abs{J}\geq \ell-k$, $h(J) > 2\eps+2\xi$ and (ii) for every set \tcolor{$J\subseteq [\ell]$} such that $\abs{J} \geq \ell-2(1+\frac{\xi}{\eps})k$, $h(J) > \eps+2\xi$.
\ET

\BPF
By~\autoref{def:ASMM}, with probability at least $1-\delta$ the value $\nu$  defined in Step~\ref{step:def:nu} of the algorithm satisfies
\begin{equation}
|\nu -\min_{\tcolor{J \subseteq [\ell]}} \{h'(J)\} | \leq  \xi \;. \label{eq:nu:guarantee}
\end{equation}

We start with proving the first item \dcolor{in~\autoref{thm:ASMC:correct}}. If there exists a set $\tcolor{J^\ast\subseteq [\ell]}$ such that $\abs{J^\ast} \geq \ell-k$ and $h(J^*)\leq \eps$,
then by Equation~\eqref{eq:nu:guarantee} and the definition of $h'$,
\[\nu \leq \min_{\tcolor{J \subseteq [\ell]}} \dcolor{\{}h'(J)\dcolor{\}} + \xi \leq h'(J^\ast) + \xi \leq  h(J^\ast)-\frac{\eps}{k}\abs{J^\ast} + \xi \leq \eps - \frac{\ell-k}{k}\eps+\xi = \left(1- \frac{\ell-k}{k}\right)\eps + \xi \;.\]
Hence, the algorithm accepts.

\dcolor{Turning to the second item in the theorem, we}
now prove that if the algorithm accepts 
(and conditioning on Equation~\eqref{eq:nu:guarantee} holding), then either (i) there exists a set $\tcolor{J^\ast\subseteq [\ell]}$ such that $\abs{J^\ast} \geq \ell-k $ and $h(J^\ast) \leq 2\eps + 2\xi$ or (ii) $\abs{J^\ast} \geq \ell-2(1+\frac{\xi}{\eps})k$ and $h(J^\ast) \leq \eps + 2\xi$.
We let $J^\ast \eqdef \arg\!\min_{\tcolor{J \subseteq [\ell]}} \{h'(J) \}$, and break the analysis into two cases, depending on $\abs{J^\ast}$.
\begin{itemize}
\item If $\abs{J^\ast}\geq \ell-k$, then since $\min_{J}\{h'(J)\} \leq \nu + \xi$  and $\nu\leq (1-\frac{\ell-k}{k})\eps +\xi$,
\[
h(J^\ast) = h'(J^\ast) + \frac{\eps}{k}\abs{J^\ast} \leq \nu + \xi + \frac{\eps}{k}\cdot \ell \leq \left(1-\frac{\ell-k}{k}\right)\eps + 2\xi + \frac{\eps \ell}{k}
= 2\eps + 2\xi,
\]
as claimed in item (i).
\item If $\abs{J^\ast}\leq \ell-k$, then
\[
h(J^\ast) = h'(J^\ast) + \frac{\eps}{k}\abs{J^\ast} \leq \nu + \xi + \frac{\eps}{k}\cdot (\ell-k) \leq \left(1-\frac{\ell-k}{k}\right)\eps + 2\xi + \frac{\eps (\ell-k)}{k}
= \eps + 2\xi.
\]

Also, since for every set $S$, $h(S)\geq 0$ and $h'(J^\ast)\leq \nu + \xi$, \dcolor{we get that}
$$\frac{\eps}{k}|J^\ast| = h(J^\ast)-h'(J^\ast) \geq \left(\frac{\ell-k}{k}-1\right)\eps-2\xi .$$ Therefore,
$\abs{J^\ast} \geq \ell-(2+\frac{2\xi}{\eps})k ,$ and item (ii) holds.
\end{itemize}\EPF

	\section{Approximate submodular function minimization}\label{sec:ASMM}
	\makeatletter{}In this section we use results from \cite{LSW:15} to obtain an approximate submodular minimization algorithm, as defined in~\autoref{def:ASMM}.
This is done in three steps: \textsf{(1)} We use the known fact that the problem of finding the minimum of a submodular function $g$ can be reduced to finding the minimum of the Lov\'asz extension for that function, denoted $\mL_g$. \textsf{(2)} We then extend the results of~\cite{LSW:15} (and specifically of Theorem 61) and provide a noisy separation oracle for $\mL_g$ when only given approximate oracle access to the function $g$. \textsf{(3)} Finally, we apply Theorem 42 from~\cite{LSW:15}, which provides an algorithm that, when given access to a separation oracle for a function, returns an approximation to that function's minimum value. \medskip

\noindent We start with the following definition of the Lov\'asz extension of a submodular function.
\BD[Lov\'asz Extension]
Given a submodular function $g\colon 2^{[\ell]} \to \R$, the \emph{Lov\'asz extension} of $g$, is a function $\mL_g\colon [0,1]^\ell\to \R$, which is defined for all $x \in [0,1]^\ell$ by
\[
\lovasz[g]{x}\eqdef \exxp{t\sim {[0,1]}} \left[ g(\setOfSuchThat{i}{ x_i\geq t}) \right]\;,
\] where $t\sim{[0,1]}$ denotes that $t$ is drawn uniformly at random from $[0,1]$.
\ED
The following theorem is standard in combinatorial optimization (see e.g.~\cite{Bach:13:submodular} and~\cite{grotschel2012geometric,schrijver2002combinatorial}) and provides useful properties of the Lov\'asz extension.
\BT \label{thm:lovasz:properties}
The Lov\'asz extension $\mL_g$ of a submodular function $g\colon 2^{[\ell]} \to \R$ satisfies the following properties.
\BE
\item $\mL_g$ is convex and $\min_{x\in [0,1]^\ell}\{\lovasz[g]{x}\}= \min_{S\subseteq [\ell]}\{g(S)\}$.
\item If $x_1\ge\ldots \geq x_\ell$ , then
\[
\lovasz[g]{x}=\sum_{i=1}^{\ell}\big(g([i])-g([i-1]) \big)x_i\;.
\]
\EE
\ET

By  the first item of~\autoref{thm:lovasz:properties}, in order to approximate the minimum value of a submodular function $g$, it suffices to approximate the minimum of its Lov\'asz extension. As discussed at the start of the section, this is done by providing a \emph{separation oracle} for $\mL_g$.

\BD [(Noisy) Separation Oracle {\cite[Definition 2]{LSW:15}}] \label{def:separation:oracle}
Let $h$ be a convex function over $\R^\ell$ and let $\Omega$ be a convex set in $\R^\ell$. A \emph{separation oracle for $h$} with respect to $\Omega$ is an algorithm that for an input $x\in\Omega$ and parameters $\eta, \gamma \geq 0$ satisfies the follows. It either asserts that $h(x)\leq \min_{y\in \Omega }\{h(y)\}+\eta$ or it outputs a halfspace $H \eqdef \{z: a^T z \leq  a^Tx + c\}$ such that
\[
\setOfSuchThat{ y\in \Omega }{ h(y)\leq h(x) }\subset H\;,
\]
where $a \in [0,1]^\ell$, $a \neq 0$,  and $c\leq \gamma \normtwo{a}$.
\ED

In Theorem 61 in~\cite{LSW:15} it is shown how to define a separation oracle for a function $g$ when given \emph{exact} query access to $g$; we adapt the proof to the case where one is only granted access to an approximate oracle for $g$, and the resulting procedure has small failure probability.

\begin{algorithm}[H]
\caption{Separation Oracle ($\oracle_g, \bar{x}, \eta,\gamma, \delta$)} \label{Alg:SepOracle}

\begin{algorithmic}[1]
	\State Assume without loss of generality that $\bar{x}_1\geq \bar{x}_2 \geq \ldots\geq \bar{x}_\ell $ (otherwise re-index the coordinates).
	\State Let $\tau=\min\{\eta/4\ell,\gamma/2\ell\}$. \label{step:tau}
	\State For each $i$, let $\tilde{g}([i])$ be the returned value from invoking $\oracle_g$ on the set $[i]$ with parameters $\frac{\tau^2}{2}$ and $\frac{\delta}{\ell}$.
	\State Define $\tilde{a} \in \R^\ell$ by $\tilde{a}_{i} \eqdef \tilde{g}([i])-\tilde{g}([i-1])$ for each $i \in [\ell]$.
	\State Let  $\tilde{\mL}_g(\bar{x}) \eqdef \tilde{a}^T\bar{x}.$
	\If {for every $i\in [\ell]$, $\abs{\tilde{a}_i}< \tau$ } \label{step:a_i:geq:tau}
	\State \textbf{return} $\bar{x}$, which satisfies ``${\mL}_g(\bar{x}) \leq \min_{y \in [0,1]^\ell} \{L_g(y)\} + \eta$''.
	\Else
	\State \textbf{return} the halfspace $H=\{z:\; \tilde{a}^T z \leq \tilde{\mL}_g(\bar{x}) + 2\tau \ell \normtwo{ \tilde{a}}\}$ \label{step:def:hyperplane}.
	\EndIf
	
\end{algorithmic}
\end{algorithm}

\BL \label{lem-SepOrCost}
Let $g\colon 2^{[\ell]}\to\R$ be a convex function, and let $\Phi_g(\cdot,\cdot)$ denote the running time of the approximate oracle \dcolor{$\oracle_g$} for $g$. For every $x \in [0,1]^\ell$, $\eta,\gamma,\delta \in (0,1)$, with probability at least $1-\delta$,~\autoref{Alg:SepOracle} satisfies the guarantees of a separation oracle for $\mathcal{L}_g$ (with respect to $[0,1]^\ell$). The algorithm makes $\ell$ queries to $\oracle_g$ with parameters $\tau^2/2$ and $\delta/\ell$, where $\tau=\min\{\eta/4\ell,\gamma/2\ell\}$, and its running time is $\ell \cdot \left(\Phi_g(\frac{\tau^2}{2},\delta/\ell)+\log \ell\right)$.
\EL

In order to prove the above lemma we will use the following theorem from~\cite{LSW:15}.
\BT[{\cite[Theorem 61]{LSW:15}}, restated] \label{thm:cite:LSW:thm61}
Let $g\colon 2^\ell \to \R$ be a submodular function. For every $x\in [0,1]^\ell$,
\[
\sum_{i=1}^\ell \big(g([i])-g([i-1])\big)x_i \leq \mL_g(x) \;.
\]
\ET

\BPFOF{\autoref{lem-SepOrCost}}
For every $i \in [\ell]$, let $a_i \eqdef g([i])-g([i-1])$, and note that by a union bound over all $i\in[\ell]$, we have that  $\max_{i\in[\ell]} \dcolor{\{}\lvert g([i]) - \tilde{g}([i]) \rvert\dcolor{\}} \leq \tau^2/2$, with probability at least $1-\delta$.
We henceforth condition on this, and observe that this implies that, for any $y\in[0,1]^\ell$,
\begin{equation}
\lvert \tilde{a}^Ty - a^Ty \rvert \leq 2\ell\cdot \frac{\tau^2}{2} = \ell \tau^2\;. \label{eq:a_tilde:approx:a}
\end{equation}
\medskip
We next consider two cases. Assume first that there exists an index $i\in[\ell]$ such that $\abs{\tilde{a}_i}\geq\tau$. That is, assume that the condition in Step~\ref{step:a_i:geq:tau} of the algorithm does not hold.  Then we prove that for every  $y \in [0,1]^\ell$ such that $\mL_g(y) \leq \mL_g(\bar{x})$ it holds that $y \in H$, where $H$ is the halfspace defined in Step~\ref{step:def:hyperplane} of the algorithm.

By~\autoref{thm:cite:LSW:thm61}, we have that $\sum_{i=1}^\ell a_i \cdot y_i \leq \mL_g(y)$ for \dcolor{every}  $y\in [0,1]^\ell$. Since $\lovasz{y}\leq \lovasz{\bar{x}}$, we get that
\begin{equation}
\tilde{a}^Ty\leq {a}^Ty + {\ell\tau^2} \leq \lovasz{y} +{\ell \tau^2} \leq \lovasz{\bar{x}}+{\ell \tau^2}\;. \label{eqn:upper:bound:a:tilde:y}
\end{equation}
By~\autoref{thm:lovasz:properties}, together with the assumption that the coordinates of $\bar{x}$ are sorted,
\begin{equation} \label{eqn:ub:mLgx}
\mL_g(\bar{x}) = \sum_{i=1}^\ell a_i \cdot \bar{x}_i \leq \sum_{i=1}^\ell \tilde{a}_i \cdot \bar{x}_i + \ell\tau^2 = \tilde{\mL}_g(\bar{x}) + \ell \tau^2 .
\end{equation}
Combining Equation~\eqref{eqn:upper:bound:a:tilde:y} and Equation~\eqref{eqn:ub:mLgx}, and since there exists an $i$ such that $|\tilde{a}_i |\geq \tau$,
\[\tilde{a}^Ty \leq \tilde{\mL}_g(\bar{x}) +2\ell \tau^2 \leq \tilde{\mL}_g(\bar{x})+2\ell\tau \normtwo{\tilde{a}} \;.\]
This implies that $y$ is in $H$ and that for $c=\tilde{\mL}_g(\bar{x})$ and $\gamma=2\tau \ell $, $H$ fulfills the requirements \dcolor{stated} in~\autoref{def:separation:oracle}.

Now consider the case that $\abs{\tilde{a}_i}\leq \tau$ for all $i\in[\ell]$. It follows that for any $y\in [0,1]^\ell$, $-\ell \tau \le\tilde{a}^T{y}\le \ell \tau$. In particular, we have that
$-\ell \tau\leq \tilde{\mathcal{L}}_g(\bar{x})\le \ell\tau$, which implies that for every $y\in[0,1]^\ell$,
\[
\tilde{\mL}_g(\bar{x})-2\ell \tau \;\leq\; -\ell\tau \;\leq\; \tilde{a}^T{y}\;.
\]
Therefore, for every $y\in[0,1]^\ell$ we get
\[\tilde{\mL}_g(\bar{x}) -3\ell\tau \;\leq\; \tilde{a}^Ty -\ell\tau \;\leq\; a^Ty \;\leq\; \mL_g(y)\;,  \]
where the second inequality follows from Equation~\eqref{eq:a_tilde:approx:a}, and the last inequality follows from~\autoref{thm:cite:LSW:thm61}. Hence, if we let $x^\ast=\arg\min_{\dcolor{x}}\{\lovasz{x}\}$, we have that
\[ \tilde{\mL}_g(\bar{x})\;\leq\; \mL_g(x^\ast)+ 3\ell \tau \;.\]
By  Equation~\eqref{eqn:ub:mLgx} we have that  $\mL_g(\bar{x}) \leq \tilde{\mL}_g(\bar{x}) + \ell \tau^2 $.
Hence,
$$\mL_g(\bar{x}) \leq \mL_g(x^*)+3\ell\tau + \ell\tau^2 \leq \mL_g(x^*)+4\ell\tau\;,$$
and since by the setting of $\tau$ in Step~\ref{step:tau} of the algorithm, $\tau\leq \eta/4\ell$, we get that  $\bar{x}$  satisfies \[ \mL_g(\bar{x})\;\leq\; \min_{y \in [0,1]^\ell}\{\mL_g(y)\} + \eta. \]
Therefore, with probability at least $1-\delta$ the algorithm satisfies the conditions of a separation oracle with parameters $\eta$ and $\gamma$.

The algorithm performs $\ell$ queries to the approximate oracle for $g$ with parameters ${\tau^2}/{2}$ and ${\delta}/{\ell}$, where $\tau= \min\{\eta/4\ell,\gamma/2\ell\}$.
Hence, the running time of the algorithm is $\ell\cdot \Phi_g(\frac{\tau^2}{2},\frac{\delta}{\ell})+\ell\log \ell$, as it also sorts the coordinates of $\bar{x}$ (in order to re-index the coordinates).
\EPFOF

\medskip
We can now use the separation oracle for $\mL_g$ and apply the following theorem to get an approximate minimum of $\mL_g$, which is also an approximate minimum of $g$.

\BT [{\cite[Theorem 42]{LSW:15}, restated}] \label{thm-Cuttingplane}Let $h$ be a convex function on $\R^\ell$ and let $\Omega$ be a convex set with constant min-width\footnote{For a compact set $K\subseteq\R^\ell$, the min-width is defined as $\min_{a\in\R^\ell\colon \normtwo{a}=1} \max_{x,y\in K} \dcolor{\{}\dotprod{a}{x-y}\dcolor{\}}$.~\cite[Definition 41]{LSW:15}. In particular, it is not hard to see that the set $K=[0,1]^\ell\subseteq B_\infty(1)$ has unit min-width.} that contains a minimizer of $h$. Suppose we have  a separation oracle for $h$ and that $\Omega$ is contained inside $B_\infty(R)\eqdef\setOfSuchThat{x }{ \norminf{x}\leq R }$, where $R>0$ is a constant. Then there is an algorithm, which for any $0<\alpha<1$ and $\eta \dcolor{ > 0}$ outputs $x\in \R^\ell$ such that
\[
h(x)-\min_{y\in \Omega}\{h(y)\}\leq \eta+\alpha\cdot\left( \max_{y\in \Omega}\{h(y)\}-\min_{y\in\Omega}\{h(y)\}\right)\;.
\]
In expectation, the algorithm performs $O\left(\ell\cdot \log\left(\frac{\ell}{\alpha}\right) \right)$ calls to~\autoref{Alg:SepOracle}, and has expected running time of
\[
O\left(\ell\cdot \operatorname*{SO}(\eta, \gamma)\log\left(\frac{\ell}{\alpha}\right)+\ell^3\log^{O(1)}\left(\frac{\ell}{\alpha}\right)\right)\;,
\]
where $\gamma=\Theta\left(\frac{\alpha}{\ell^{3/2}}\right)$ and $\operatorname*{SO}(\eta, \gamma)$ denotes the running time of the separation oracle when
invoked with parameters $\eta$ and $\gamma$.
\ET

\BC\label{coro:submodular:noisy:minimization}
Let $g\colon 2^{[\ell]} \to \R$ be a submodular function. There exists an algorithm that, when given access to $\oracle_g$, and for input  parameters $\xi, \delta\in(0,1)$,  returns with probability at least $9/10-2\delta$ a value $\nu \in \R$ such that $\nu\leq \min_{S \subseteq [\ell]} \{g(S)\} + \xi$.

The algorithm performs $\ell\log\left(\frac{\ell M}{\xi}\right)$ calls to $\oracle_g$ with parameters $\frac{\xi^2}{128\ell^5M^2}$ and $\frac{\delta}{C\ell^2\log\left(\frac{\ell M}{\xi} \right)}$, where $M\eqdef\max\left\{2\max_{S \subseteq [\ell]}\{\abs{g(S)}\},\xi/2 \right\}$ and $C>0$ is an absolute constant.
The running time of the algorithm is
\[
O\!\left( \ell^2\cdot \Phi_g\left( \frac{\xi^2}{128\ell^{5}M^2}, \frac{\delta}{C\ell^2\log\frac{\ell M}{\xi}} \right)\log{\frac{\ell M}{\xi}}  \tcolor{+\ell^2\log \ell} + \ell^3 \log^{O(1)}\frac{\ell M}{\xi}\right) \;,
\]
where $\Phi_g$ is the running time of $\oracle_g$.
\EC
\BPF
\sloppy
We refer to the algorithm from~\autoref{thm-Cuttingplane} as {\em the minimization algorithm\/} and apply it to $\mL_g$, with~\autoref{Alg:SepOracle} as a separation oracle.
Once the minimization algorithm returns a point ${x} \in [0,1]^\ell$, we return the value $\nu=\oracle_{\mL_g}(x, \xi/4, \delta)$.

Let $M'\eqdef 2\max_{S \subseteq [\ell]}\{\abs{g(S)}\}$, and recall that $\mL_g(x) = \exxp{t \sim[0,1]}{[g(\setOfSuchThat{i }{ x_i \geq t }) ]}$. Hence, $\max_{x \in [0,1]^\ell} \{\mL_g(x)\} - \min_{x \in [0,1]^\ell} \{\mL_g(x)\} \leq M'$. Setting $\alpha<\xi/(2M)$ and $\eta = \xi/4$ ensures that $0 <\alpha <1$ and that
\begin{equation} \label{eqn:eta:alpha}
\eta + \alpha\cdot\left(\max_{x \in [0,1]^\ell} \{\mL_g(x)\} - \min_{x \in [0,1]^\ell} \{\mL_g(x)\} \right) \leq  \eta +\alpha M' \leq 3\xi/4 \;.
\end{equation}

The minimization algorithm invokes the separation oracle $C_1 \cdot \ell \log (\ell/\alpha) = C_1 \cdot \ell \log (\ell M/\xi)$ times in expectation, for some constant $C_1$. If at some point the number of calls to the separation oracle exceeds $ 10C_1 \cdot \ell \log (\ell M/\xi)$, then we halt and return \fail. By Markov's inequality this happens with probability at most $1/10$.
Hence, every time the minimization algorithm calls the separation oracle with parameters $\eta$ and $\gamma$ we invoke~\autoref{Alg:SepOracle} with parameters $\eta, \gamma$ and  \tcolor{$\delta' = \frac{\delta}{10C_1 \ell \log\left(\frac{\ell M}{\xi}\right) }$}. Therefore, with probability at least $1-1/10-\delta$ all the calls to~\autoref{Alg:SepOracle} satisfy the guarantee of a separation oracle for $\mL_g$ with parameters $\eta$ and $\gamma$. By~\autoref{thm-Cuttingplane} and Equation~\eqref{eqn:eta:alpha}, with probability at least $9/10-\delta$ the minimization algorithm returns a point $x$ such that
\[\mL_g(x) - \min_{{y}\in [0,1]^\ell} \{\mL_g({y})\}\leq \eta +  \alpha\cdot\left(\max_{y \in [0,1]^\ell} \{\mL_g(y)\} - \min_{y\in [0,1]^\ell} \{\mL_g(y)\} \right)\le \frac{3\xi}{4} \;,\]
and with probability at least  $9/10-2\delta$ the value $\nu$ satisfies
\[
\nu \leq \min_{y \in [0,1]^\ell} \{\mL_g(y)\} + \xi \;,
\]
as desired.

By the above settings and by~\autoref{lem-SepOrCost} we get that $\tau = \frac{\xi}{8\ell^{5/2}M}$ so the running time of each invocation of the separation oracle is

\[
\ell\cdot \Phi_g\left( \frac{\tau^2}{2}, \frac{\delta'}{\ell} \right)+\tcolor{\ell \log \ell}
= \ell\left( \Phi_g\left( \frac{\xi^2}{128^{5}M^2}, \frac{\delta}{10C_1 \ell^2\log\frac{\ell M}{\xi}} \right) +\tcolor{\log \ell}\right)\;.
\]
Since the evaluation of $\nu$ in the final step is negligible in the running time of the minimization algorithm, we get that the overall time complexity is
\[
O\!\left( \ell^2\cdot \Phi_g\left( \frac{\xi^2}{128\ell^{5}M^2}, \frac{\delta}{10C_1\ell^2\log\frac{\ell M}{\xi}} \right)\log{\frac{\ell M}{\xi}} \tcolor{+ \ell^2\log \ell} + \ell^3 \log^{O(1)}\frac{\ell M}{\xi}\right) \;.
\]

\EPF

\BC\label{coro:submodular:noisy:minimization:junta}
There exists an algorithm that, when given query access to a function $f\colon \bool^n\to\bool$ and a partition $\mI=\{I_1,\dots,I_\ell\}$ of $[n]$ into $\ell$ parts, as well as input parameters $k\in\N, \eps, \xi\in(0,1)$, satisfies the following. It has time and query complexity $\tilde{O}\left( \max\left( \frac{\ell^{12}}{\xi^4}, \frac{\ell^{16}\eps^4}{k^4\xi^4}\right) \right)$, and distinguishes with probability at least $5/6$ between the 
\dcolor{following two}
cases:
\BE
\item There exists a set $S\subseteq[\ell]$ such that $\abs{S} \geq \ell-k$ and $h(S)\leq \eps$.
\item For every set $S$ such that $\abs{S} \geq \ell-2(1+\frac{\xi}{\eps})k$, $h(S) > 2(\eps+\xi)$
\EE
where $h\colon 2^\ell \to\R$ is defined as $h(S)\eqdef \infl[f]{\phi_{\mI}(S)}$. \smallskip

\noindent Moreover, the second item can be strengthened so that it holds for functions $f$ \dcolor{that} satisfy the following: (i) for every set $S$ such that  $\abs{S}\geq \ell-k$, $h(S) > 2\eps+2\xi$ and (ii) for every set $S$ such that $\abs{S} \geq \ell-2(1+\frac{\xi}{\eps})k$, $h(S) > \eps+2\xi$. 
\EC

\BPF
We apply~\autoref{coro:submodular:noisy:minimization} to $h'\colon2^{[\ell]}\to\R$, defined as in~\autoref{alg:ASMC} by $h'(S) \eqdef h(S) - \frac{\eps}{k}\abs{S}$, with $\xi$, $M\eqdef \max\left( 2\max(2, \frac{\eps \ell}{k}) , \xi/2 \right)=4\max(1, \frac{\eps \ell}{2k})$, and $\delta\eqdef \frac{1}{30}$. In order to do so, we need to simulate
\dcolor{an approximate oracle for $h'$ (as defined in Definition~\ref{def:approx-oracle}).}
 Since $h(S) = \infl[f]{\tcolor{\phi_{\mI}(S)}}$, in order to estimate $h'(S)$ \dcolor{within} an additive \dcolor{approximation of} ${\tau'}$ with probability at least $1-\delta'$, it is sufficient to estimate $\infl[f]{\tcolor{\phi_{\mI}(S)}}\in[0,2]$ \dcolor{within} an additive \dcolor{approximation of}  ${\tau'}$ with probability at least $1-\delta'$ (indeed, the additional term $\frac{\eps}{k}\abs{S}$ can be computed exactly). By Chernoff bounds, this can be done with $\Phi_h(\tau',\delta')= O(\frac{1}{\tau'^2}\log\frac{1}{\delta'})$ queries to $f$.

This yields an approximate oracle $\oracle_h$, and therefore $\oracle_{h'}$ (with success probability $9/10-2\delta = 5/6$) which can be provided to the algorithm of~\autoref{thm:ASMC:correct}. \dcolor{The resulting} query complexity is
\[
O\!\left( \ell^2\cdot \Phi_h\left( \frac{\xi^2}{\ell^{5}M^2}, \frac{1}{10C_1\ell^2\log\frac{\ell M}{\xi}} \right)\log{\frac{\ell M}{\xi}}  \tcolor{+ \ell^2\log \ell} + \ell^3 \log^{O(1)}\frac{\ell M}{\xi}\right)
\]
which, given the above expression for $\Phi_h$, can be \dcolor{bounded} as follows.

\begin{itemize}
\item If $\eps < \frac{2k}{\ell}$, so that $M = 4$, this simplifies as
\[
O\!\left( \frac{\ell^{12}}{\xi^4}\log^2{\frac{\ell}{\xi}}\right).
\]
\item If $\eps \geq \frac{2k}{\ell}$, which implies that $M = \frac{2\eps \ell}{k}$, this becomes
\[
O\!\left( \frac{\ell^{16}\eps^4}{k^4\xi^4}\log^2 \ell \right).
\]
\end{itemize}
Observing that the function $h$ is indeed a non-negative submodular function (and that $h'$ is also submodular \dcolor{since it is} the sum of a submodular
 \dcolor{function} and a modular function) allows us to conclude by~\autoref{thm:ASMC:correct}.
\EPF

\noindent In particular, setting $\xi=\eps$ we get the following:
\BC\label{coro:testing:kpart:juntas:parameterized}
There exists an algorithm that, given query access to a function $f\colon\bool^n \to \bool$, a fixed partition $\mI$ of $[n]$ into $\ell= O(k^2)$ parts,
and parameters $k\geq 1$ and $\eps \in (0,1)$, satisfies the following. The query	complexity of the algorithm is
$
\tilde{O}\left(\frac{k^{24}}{\eps^4}+k^{28}\right) = \poly(k,1/\eps)
$, and:
\BE
\item if $f$ $\frac{\eps}{2}$-approximates being a $k$-part junta with respect to $\mI$, then the algorithm accepts
with probability at least $\frac{5}{6}$;
\item if $f$ $2\eps$-violates being a $4k$-part junta  with respect to $\mI$, then the algorithm rejects
with probability at least $\frac{5}{6}$.
\EE
Moreover, the second item can be strengthened to ``simultaneously $2\eps$-violates being a $k$-part junta and $\frac{3}{2}\eps$-violates being a $4k$-part junta.''
\EC
\BPF
Follows immediately from applying~\autoref{coro:submodular:noisy:minimization:junta} with $\xi=\eps$.
\EPF\vspace{5pt}

The tolerant junta testing theorem (\autoref{theo:tol:testing:juntas:relaxed}) follows immediately from the above, together with~\autoref{prop:tol:testing:juntas:reduction:n:m}. With probability at least $5/6$, a random partition of the variables into $\ell \eqdef 192k^2$ parts will have the right guarantees, reducing the problem to distinguishing between $\frac{\eps}{2}$-approximating being a $k$-part junta vs. $2\eps$-violating being a $4k$-part junta (with regard to this random partition). Overall, by a union bound, the result is therefore correct with probability at least $2/3$.

	\section{A tradeoff between tolerance and query complexity}\label{sec:tradeoff}
	\makeatletter{}In this section, we show how to obtain a smooth tradeoff between the amount of tolerance and the query complexity. Formally, we prove~\autoref{theo:tol:testing:juntas:tradeoff}, restated below.

\testingtradeoff*

Before delving into the proof of the theorem,
 we discuss some of its consequences. Setting $\rho=\bigOmega{1}$, we obtain a tolerant tester that distinguishes {between functions $\bigO{\eps}$-close to $\junta$ and functions $\eps$-far from $\junta$,} with query complexity ${2^{\bigO{k}}}/{\eps}$ -- thus matching (and even improving) the simple tester described in~\autoref{sec:basics}. At the other end of the spectrum, setting $\rho=\bigO{{1}/{k}}$ yields a weakly tolerant tester that distinguishes $\bigO{{\eps}/{k}}$-close {to $\junta$ from $\eps$-far from $\junta$}, but with query complexity $\tildeO{{k^2}/{\eps}}$ -- qualitatively matching the guarantees provided by the junta tester of~\cite{FKRSS:04}.
       
\subsection{Useful bounds on the expected influence of a random $\rho$-subset of a set}
\tcolor{In this subsection we formally define the $\rho$-subset influence of a set and prove that for every set $J \subseteq [\ell]$, its $\rho$-subset influence is at least $\frac{\rho }{3}\cdot\infl{\phi_{\mI}(J)}$ and at most $\infl{\phi_{\mI}(J)}$. Then in the next subsection we provide an algorithm that simultaneously estimates the $\rho$-subset influence of all subsets $J$ of $\ell$ of size $\ell -k$. The running time of the algorithm is $\bigO{ \frac{k\log k}{\eps\rho(1-\rho)^k} }$.
}
We start with a few definitions and notations.

\BD
For any $\rho\in(0,1)$ and any set $R$, we denote by $S \sim_\rho R$ the random $\rho$-biased subset of $R$, resulting from including independently each $i \in R$ in $S$ with probability $\rho$. We refer to such a set $S$ as a \emph{random $\rho$-subset of $R$}.
\ED

\BD For a partition $\mI=\{I_1, \ldots, I_\ell\}$ and a set $J \subseteq [\ell]$ we refer to the expected value of the influence of a random $\rho$-biased subset of $J$, $\shortexpect_{S\sim _\rho J} [ \infl{\phi_{\mI}(S)} ]$, as 
\dcolor{the} \emph{$\rho$-subset influence} \dcolor{of $J$ (with respect to $\mI$)}.
\ED

The next lemma describes the connection between the influence of a set $J$ and its $\rho$-subset influence.

\BL\label{lemma:infl:rhosubset}
Let  $\mI=\{I_1, \ldots, I_\ell\}$ be a partition of $[n]$. Then, for any $J\subseteq [\ell]$,
\[
\frac{\rho}{3}\infl{\phi_{\mI}(J)} \leq \shortexpect_{S\sim _\rho J} [ \infl{\phi_{\mI}(S)} ] \leq \infl{\phi_{\mI}(J)}.
\]
\EL

\BPF
The upper bound is immediate by monotonicity of the influence, as $\infl{\phi_{\mI}(S)} \leq \infl{\phi_{\mI}(J)}$ for all $S\subseteq {J}$.
As for the lower bound, let $j=|J|$ and observe that
\begin{equation}
\shortexpect_{S\sim _\rho J } [ \infl{\phi_{\mI}(S)}  ]  = \sum_{s=1}^{j} \sum_{S\subseteq {J} : \abs{S}=s } \rho^{\new{s}} (1-\rho)^{j - s } \cdot \infl{\phi_{\mI}(S)}  \label{eq:exp_infl}.
\end{equation}
We will lower bound the sum $\sum_{S\subseteq {J} : \abs{S}=s }   \infl{\phi_{\mI}(S)}$ for each $s$ separately.
In order to do so we define a \emph{legal collection of covers for a set $J$}:

\BD
Let $J$ be a set of $j$ elements, and for any $s \in [j]$ consider the  family $\subsetcoll{J}{s}$ of
all $\binom{j}{s}$ subsets of $J$ of size $s$. We shall say that $\mathcal{C} \subseteq \subsetcoll{J}{s}$ is a \emph{cover} of $J$
if $\bigcup_{Y\in \mathcal{C}} Y = J$.  We shall say that a collection of covers $\mathscr{C}_J= \{\mathcal{C}_1,\dots,\mathcal{C}_r\}$ is a \emph{legal collection of covers for $J$} if each $\mathcal{C}_t \in \mathscr{C}_J$ is a cover of $J$ and these covers are disjoint.
\ED
Thus, we are interested in showing that there exists a legal collection of covers for ${J}$ whose size $m$ is ``as big as possible.'' This is what the next claim guarantees, establishing that there exists such a cover achieving the optimal size:
\begin{claim}\label{claim:legal:collection:bound}
	For any set $J$ of $j$ elements, there exists a legal collection of covers $\mathscr{C}_J$ for $J$ of size at least
	\[
	|\mathscr{C}_J| \geq \flr{\frac{{\binom{j}{s}} }{ \clg{\frac{j}{s}} }}.
	\]
	(Moreover, this bound is tight.)
\end{claim}
\dcolor{\autoref{claim:legal:collection:bound}} follows from a result due to Baranyai~\cite{Baranyai:75} on factorization of regular hypergraphs: for completeness, we state this result, and describe how to derive the claim from it, in~\autoref{app:baranyai}. Observe that if $s$ divides $j$ then $\flr{\frac{ {\binom{j}{s}} }{ \clg{\frac {j}{s}} }} = \binom{j-1}{s-1}$; and otherwise
\[
\flr{\frac{ {\binom{j}{s}} }{ \clg{\frac {j}{s}} }}
= \flr{\frac{ \frac {j}{s} }{ \clg{\frac {j}{s}} }\binom{j-1}{s-1}  }
\geq \flr{\frac{ \frac {j}{s} }{ \frac {j}{s}+1 }\binom{j-1}{s-1}  }
\geq \flr{\frac{ 1 }{ 2 }\binom{j-1}{s-1}  }
\geq  \frac{1}{3} \binom{j-1}{s-1}
\;.\]
Therefore,
\begin{align*}
	\sum_{S\subseteq {J} : \abs{S}=s }   \infl{\phi_{\mI}(S)} &= \sum_{S \in \dcolor{\subsetcoll{J}{s}} }   \infl{\phi_{\mI}(S)} \geq  \sum_{\mathcal{C} \in \mathscr{C}_J} \sum_{S \in \mathcal{C}} \infl{\phi_{\mI}(S)} \\
	&=\abs{ \mathscr{C}_J }\cdot \infl{\phi_{\mI}(J)}  \geq \frac{1}{3}\binom{j-1}{s-1} \infl{\phi_{\mI}(J)}.
\end{align*}
Plugging the above into Equation~\eqref{eq:exp_infl}, \dcolor{we obtain} that
\begin{align*}
	\infl{\phi_{\mI}(J)} &\geq \sum_{s=1}^{j}  \rho^{\new{s}}(1-\rho)^{j-s} \cdot \left(\frac{1}{3}\binom{j-1}{s-1} \infl{\phi_{\mI}(J)} \right) \\
	&= \frac{\rho}{3} \infl{\phi_{\mI}(J)}\sum_{s=1}^{j} \binom{j-1}{s-1} \rho^{s-1}(1-\rho)^{j-s} \\
	&= \frac{\rho}{3} \infl{\phi_{\mI}(J) } (\rho+(1-\rho))^{j-1}=\frac{\rho}{3} \infl{\phi_{\mI}(J)},
\end{align*}
which concludes the proof.
\EPF

\subsection{Approximation of the $\rho$-subset influences}\label{ssec:algo:estimate:infl:rhobiasedset}

We now describe and analyze an algorithm that given a partition $\mI=\{I_1, \ldots, I_\ell\}$, allows to simultaneously get good estimates of the $\rho$-subset influences of all subset $J \in \subsetcoll{[\ell]}{\ell -k}$. This algorithm is the main building block of the tolerant junta tester of~\autoref{theo:tol:testing:juntas:tradeoff}.

\begin{algorithm}[H]
	\caption{Simultaneously Approximate $\rho$-subset Influence  ($\oraclequery{f},\rho,\eps,\gamma,k,\ell,\mI$)} \label{algo:approx.NoisyInf}
	\begin{algorithmic}[1]
				\State Set $m=\frac{C\cdot k\log \ell}{\gamma^2\eps\rho(1-\rho)^k}$, where $C\geq 1$ is an absolute constant. \Comment{$C\geq 256\ln 2$ is sufficient.}
		\For {$i=1$ to $m$}
		\State  Let $S^i \sim_\rho [\ell]$. \label{ChooseS}
		\State Pick $x^i \in \bool^{n}$ uniformly at random, and let $z^i\overset{}{\sim}  x^i_{\phi_{\mI}(S^i)}$.
		\State Set $y \gets x^i_{\phi_{\mI}(\bar{S}^i)}\sqcup z^i$.
		\State Set $\vartheta_{S^i} \gets \indicSet{\{ f(x^i)\neq f(y^i) \}}$ \;. \label{PickY}
		\EndFor
		\State Let  \dcolor{$\mS$ be the multiset of subsets $S^1,\dots,S^m$.} 		\For {every $J \in \subsetcoll{[\ell]}{\ell -k}$}\label{Approx-start}
		\State Let $\mS_{J}\subseteq \mS$ denote the subset of sets $S \dcolor{\in \mS}$ such that $S \subseteq{J}$\;. \label{mS_J}
		\State Let $\dcolor{\nu}^{\rho}_{{J}}\gets \frac{1}{\abs{\mS_J}}\sum\limits_{S \in \mS_J}\vartheta_{S}$\label{Approx-end} be the estimate of the $\rho$-subset influence of $J$. \label{step:nu_rho_J}
				\EndFor		
	\end{algorithmic}
\end{algorithm}

\BL \label{lem:BlackboxBiasedAprox}
Let $\mI=\{ I_1,\dots, I_\ell \}$ be a partition of $[n]$. For every $\eps\in(0,1)$ and \new{$\rho \in (0,1)$},~\autoref{algo:approx.NoisyInf} satisfies that, with probability at least $1-o(1)$, the following holds simultaneously for all sets $J \in \subsetcoll{[\ell]}{\ell -k}$ such that $\abs{J} =\ell-k$:
\begin{enumerate}
	\item if $\shortexpect_{S\sim _\rho J} [ \infl{\phi_{\mI}(S)} ] > \frac{\rho\eps}{3}$, \new{then} the estimate $\dcolor{\nu}^\rho_J$  is within a multiplicative factor of $(1\pm\gamma)$ of the $\rho$-subset influence of $J$.
		\item if  $\shortexpect_{S\sim _\rho J} [ \infl{\phi_{\mI}(S)} ]  \leq \frac{\rho\eps}{4}$, \new{then} the estimate $\dcolor{\nu}^\rho_J$ does not exceed $(1+\gamma)\frac{\rho\eps}{4}$.
\end{enumerate}
\EL

\BPF 
\tcolor{Let $m'\eqdef \frac{1}{2}(1-\rho)^k\cdot m=\frac{Ck\log \ell}{2\gamma^2\epsilon\rho}$}. We first claim that for any fixed set $J\subseteq [\ell]$ of size $\ell-k$, with probability at least $1-o(\ell^{-2k})$, $\abs{\mS_J} \geq m'$. To see why this is true, fix some $J\subset[\ell]$ of size  $\ell-k$. For every $i\in[m]$, let $\indic{S_i\subseteq {J} }$ be an indicator variable which is equal to $1$ if and only if $S_i\subseteq {J}$. Then, for every $i\in[m]$, $\probaOf{ \indic{S_i\subseteq {J} }=1 }=  (1-\rho)^{k}$.
By a Chernoff bound,
\[
\probaOf{ \frac{1}{m}\sum_{i=1}^{m}\indic{S_i\subseteq {J} } < \frac{1}{2}\cdot (1-\rho)^{k} }
\leq e^{-\frac{m}{8} (1-\rho)^{k} } = e^{- \frac{C\cdot k\log \ell}{8\eps\rho\gamma^2} }
< 2^{-4{k\log \ell} } \;,
\]
for a suitable choice of $C\geq 1$. Therefore, by a union bound over all $ \verynew{\binom{\ell}{\ell-k} =\binom{\ell}{k}= 2^{(1+o(1))k\log \ell}}$ sets $J \in \subsetcoll{[\ell]}{\ell -k}$,  it holds that with probability $1-o(1)$, for every such $J$, $\abs{\mS_J} \geq m'$. We hereafter condition on this. \medskip

We now turn to prove the two items of the lemma. \tcolor{Let $X=\{x^1, \ldots, x^m\}$ and $Z=\{z^1, \ldots, z^m \}.$ For a set $S^i$, $\shortexpect_{x^i,z^i}[\vartheta_S] = \infl{\phi_{\mI}(S)}$. Hence,} by the definition of $\nu^\rho_J$ in Step~\ref{step:nu_rho_J} of the algorithm,
 \begin{align*} \label{eq:exp_nu_J}
\shortexpect[\dcolor{\nu}^\rho_J] &= \shortexpect_{\mS\tcolor{,X,Z}}\left[\frac{1}{\abs{\mS_J}} \sum_{S \in \mS_J} \vartheta_S \right]
\tcolor{=\shortexpect_{\mS}\left[\frac{1}{\abs{\mS_J}} \sum_{S^i \in \mS_J} \shortexpect_{x^i,y^i}[\vartheta_{S^i}] \right]}
\\ &=
\sum_{S \subseteq J}\Pr[S \in \mS] \cdot \infl{\phi_{\mI}(S)}=
\shortexpect_{S \sim_\rho J}\left[ \infl{\phi_{\mI}(S)}\right] . \numberthis
\end{align*}
Consider a set $J$ with $\shortexpect_{S\sim _\rho J} [ \infl{\phi_{\mI}(S)} ]  > \frac{\rho\eps}{3}$.
 By Equation~\eqref{eq:exp_nu_J}, $\shortexpect[\dcolor{\nu}^\rho_J] > \frac{\rho\eps}{3}$. Therefore, by a Chernoff bound, and  since for every $J$, $\abs{\mS_J} \geq m'=\frac{C k \log \ell}{2\gamma^2 \eps \rho}$,
\begin{eqnarray*}
\probaOf{ \abs{ 
				\tcolor{\nu^{\rho}_J}-  \rhosubsetinf } > \gamma \rhosubsetinf  }
&\leq&  2e^{-\frac{\verynew{\abs{\mS _J}} \gamma^2 \cdot \rhosubsetinf}{3}} \\
\leq  2e^{-\frac{m'\gamma^2\eps\rho}{9}}
< 2^{-4k\log \ell},
\end{eqnarray*}
\new{again} for a suitable choice of the constant $C\geq 1$.
By taking a union bound over all subsets $J \in \subsetcoll{[\ell]}{\ell -k}$, we get that, with probability at least $1-o(1)$, for every ${J}$ such that  $\rhosubsetinf > \frac{\rho\eps}{3}$, it holds that $\dcolor{\nu}\tcolor{^{\rho}_J} \in (1\pm \gamma)\cdot \rhosubsetinf$.

Now consider a set $J\subseteq[\ell]$ such that $\abs{\bar{J}} > \ell-k$ and \new{$\shortexpect_{S\sim _\rho J} [ \infl{\phi_{\mI}(S)} ] \leq \frac{\rho\eps}{4}$}. By a  Chernoff bound:
\begin{align*}
\probaOf{ \tcolor{\nu^\rho_J} > (1+\gamma)\frac{\rho\eps}{\tcolor{4}} }
&\leq  e^{-\frac{\gamma^2}{3} \frac{\rho\eps}{4}\verynew{\abs{\mS_J}} }
\leq  e^{-\frac{\gamma^2\rho\eps}{12}m' }
< 2^{-4k\log \ell} .
\end{align*}
The claim follows by taking a union bound over all subsets $J \in \subsetcoll{[\ell]}{\ell -k}$ for which \new{$\shortexpect_{S\sim _\rho J} [ \infl{\phi_{\mI}(S)} ] \leq \rho\eps/4$}. Overall, the conclusions above hold with probability at least $1-o(1)$, as stated.\smallskip
\EPF

\subsection{Tradeoff between tolerance and query complexity}\label{sec:tolerant:testing:juntas:noisy}

\medskip
We now describe how the algorithm from the previous section lets us easily derive the tolerant tester of~\autoref{theo:tol:testing:juntas:tradeoff}.

\begin{algorithm}[H] \label{algo:tradeoff:noisy:influence}
	\caption{$\rho$-Tolerant Junta Tester ($\oraclequery{f},\eps,\rho,k$)}\label{algo:junta:tradeoff}
	\begin{algorithmic}[1]
		\State Create a random partition $\mI$ of $\ell=24k^2$ parts by uniformly and independently assigning each coordinate to a part.
		\State Run~\autoref{algo:approx.NoisyInf} with the partition $\mI$, $\ell=24k^2$ and $\gamma=1/8$.
		\If {there is a set $J\subset [\ell]$ of size $\ell-k$ such that $\dcolor{\nu}^{\rho}_J \le \frac{9\rho\epsilon}{32}$}
		\State\label{algo:tradeoff:noisy:influence:accept} \Return \accept.
		\EndIf
		\State \Return \reject.
	\end{algorithmic}
\end{algorithm}

\BPFOF {\autoref{theo:tol:testing:juntas:tradeoff}}  Given~\autoref{prop:tol:testing:juntas:reduction:n:m} it is sufficient to consider a partition $\mI$ of size $\ell = 24k^2$ and show that~\autoref{algo:junta:tradeoff} distinguishes with probability at least $5/6$ between the following two cases.
\begin{enumerate}
	\item $f$ $\frac{\rho\eps}{8}$-approximates being a $k$-part junta with respect to $\mI$;
	\item $f$ $\frac{\eps}{2}$-violates being a $k$-part junta with respect to $\mI$
\end{enumerate}
Suppose first that $f$ $\frac{\rho\eps}{8}$-approximates being a $k$-part junta with respect to $\mI$. Then by \autoref{def:partition_juntas}, there exists a set $J \in \ellsubsets$ such that $\infl{\phi({J})} \leq \frac{\rho\eps}{4}$. By~\autoref{lemma:infl:rhosubset}, $\rhosubsetinf \leq \frac{\rho\eps}{4}$, and by~\autoref{lem:BlackboxBiasedAprox}, we have that with probability at least $1-o(1)$, the estimate $\dcolor{\nu}^\rho_{J}$ is at most $(1+1/8)\frac{\epsilon\rho}{4}\le \frac{9\eps\rho}{32}$. Therefore,~\autoref{algo:junta:tradeoff} will return \accept when considering $J$.

Consider now the case where $f$ $\frac{\eps}{2}$-violates being a $k$-part junta with respect to $\mI$. Hence, by \autoref{def:partition_juntas},
every set $J \in \ellsubsets$ is such that $\infl{\phi({J})} > \eps$,  and by~\autoref{lemma:infl:rhosubset}, we have that $\rhosubsetinf \geq \frac{\rho}{3}\infl{\phi(J)} > \frac{\rho \eps}{3}$. Therefore, by \autoref{lem:BlackboxBiasedAprox}, for every $J \in \ellsubsets$, with probability at least $1-o(1)$,
\[
\dcolor{\nu}^\rho_{{J}} \geq \frac{7}{8} \rhosubsetinf >\frac{9\rho \eps}{32}\;.
\]
Thus, with probability at least $1-o(1)$,~\autoref{algo:junta:tradeoff} will reject $f$.
\EPFOF

	\section{``Instance-adaptive'' tolerant isomorphism testing}\label{sec:isomorphism}
	\makeatletter{}In this section, we show how the machinery developed in~\autoref{sec:tradeoff}, and more precisely the algorithm from~\autoref{theo:tol:testing:juntas:tradeoff}, can be leveraged to obtain \emph{instance-adaptive tolerant isomorphism testing} between two unknown Boolean functions $f$ and $g$, as defined below.

We begin with some notation: for $f,g\colon\bool^n\to\bool$, we denote by $\distiso{f}{g}$ the distance between $f$ and the closest isomorphism of $g$, that is $\distiso{f}{g} \eqdef \min_{\pi\in\perm[n]} \dist{f}{g\circ\pi}$. Given query access to two unknown Boolean functions $f,g\colon\bool^n\to\bool$ and a parameter $\eps\in(0,1)$, isomorphism testing then amounts to distinguishing between \textsf{(i)}~$\distiso{f}{g}=0$; and \textsf{(ii)}~$\distiso{f}{g} > \eps$.\footnote{Phrased differently, this is testing the property $\property = \setOfSuchThat{(f,f\circ\pi)}{f\in 2^{2^n}, \pi\in\perm[n]} \subseteq 2^{2^n}\times 2^{2^n}$.}

\noindent Our result will be parameterized in terms of the \emph{junta degree} of the unknown functions $f$ and $g$, formally defined below:
\begin{definition}[Junta degree]
Let $f\colon\bool^n\to\bool$ be a Boolean function, and $\gamma\in[0,1]$ a parameter. We define the \emph{$\gamma$-junta degree of $f$} as the smallest integer $k$ such that $f$ is $\gamma$-close to being a $k$-junta, that is
\[
   \kclose(f,\gamma) \eqdef \min\setOfSuchThat{ k\in [n] }{ \dist{f}{\junta[k]} \leq \gamma }.
\]
Finally, we extend this definition to two functions $f,g$ by setting $\kclose(f,g,\gamma) = \min( \kclose(f,\gamma),\kclose(g,\gamma) )$.
\end{definition}

\noindent With this terminology in hand, we can restate~\autoref{theo:iso:testing:robust:kclose}:
\begin{theorem}[\autoref{theo:iso:testing:robust:kclose}, rephrased]\label{theo:iso:testing:robust:kclose:detailed}
 There exist absolute constants $c \in (0,1)$, $\eps_0\in(0,1)$ and a \verynew{tolerant} testing algorithm for isomorphism of two unknown functions $f$ and $g$ with the following guarantees. On inputs $\eps\in(0,\eps_0]$, $\delta\in(0,1]$, and query access to functions $f,g\colon\bool^n\to\bool$:
  \begin{itemize}
    \item if $\distiso{f}{g} \leq c\eps$, then it outputs \accept with probability at least $1-\delta$;
    \item if $\distiso{f}{g} > \eps$, then it outputs \reject with probability at least $1-\delta$.
  \end{itemize}
  The query complexity of the algorithm satisfies the following, where $\kclose=\kclose(f,g,\frac{\rho c\eps}{\new{16}})$ is the \newest{$\frac{\rho c\eps}{16}$}-junta degree of $f$ and $g$:
  \begin{itemize}
    \item it is \tcolor{$\tilde{O}\big(2^{\frac{\kclose}{2}} \frac{1}{\eps}\log\frac{1}{\delta}\big)$} with probability at least $1-\delta$;
    \item it is always at most \tcolor{$\tilde{O}\big(2^{\frac{n}{2}}\frac{1}{\eps}\log\frac{1}{\delta} \big)$}.
  \end{itemize}
  Moreover, one can take \new{$c=\frac{1}{1750}$}, and \new{$\eps_0\eqdef \frac{16}{15} (5-2\sqrt{6}) \simeq 0.108$}.
\end{theorem}

\subsection{Proof of~\autoref{theo:iso:testing:robust:kclose:detailed}}
As described in~\autoref{ssec:overview:techniques}, our algorithm first performs a linear search on $k$, invoking at each step the tolerant tester of~\autoref{sec:tradeoff} with parameter $\eps'$, to obtain (with high probability) a value $k^\ast$ such that $k^\ast(f,g,\eps') \leq k^\ast \leq k^\ast(f,g,\frac{\rho \eps'}{\new{16}})$. In the second stage, it calls a ``noisy sampler'' to obtain uniformly random labeled samples from the ``cores'' of the $\kclose$-juntas closest to $f$ and $g$ \dcolor{(both notions are defined formally in~\autoref{subsec:noisysamp-core})}, and robustly tests isomorphism between them. We accordingly divide this section in two, proving respectively these two statements:

\begin{lemma}\label{lemma:testing:iso:linear:search}
  There exists an algorithm (\autoref{alg:testing:iso:linear:search}) with the following guarantees. On inputs $\eps',\delta\in(0,1)$ and query access to $f,g\colon\bool^n\to\bool$, it returns a value $0\leq k \leq n$, such that:
  \begin{itemize}
    \item with probability at least $1-\delta$, we have that:
    \begin{enumerate}[(i)]
      \item $k^\ast(f,g,\eps') \leq \newest{k} \leq k^\ast(f,g,\frac{\rho \eps'}{\new{16}})$;
      \item the algorithm performs $\bigO{ 2^{\frac{k}{2} + o(k)} \cdot  \frac{1}{\eps} \log\frac{1}{\delta} }$ queries;
    \end{enumerate}
    \item the algorithm performs at most $\bigO{ 2^{\frac{n}{2} + o(n)}  \cdot \frac{1}{\eps} \log\frac{1}{\delta} }$ queries.
  \end{itemize}
\end{lemma}

\begin{proposition}\label{prop:iso:testing:robust:junta}
 There exists an algorithm (\autoref{alg:testing:iso:robust:kjunta}) with query complexity $\tildeO{\frac{2^{k/2}}{\eps}}$ for testing of isomorphism of two unknown functions $f$ and $g$, under the premise that $f$ is close to $\junta[k]$. More precisely, there exist absolute constants $c>0$ and $\eps_0\in(0,1]$ such that, on inputs $k\in\N$, $\eps\in(0,\eps_0]$ and query access to functions $f,g\colon\bool^n\to\bool$, the algorithm has the following guarantees. Conditioned on $\dist{f}{\junta[k]}\leq c\eps$, it holds that:
  \begin{itemize}
    \item if $\distiso{f}{g} \leq c\eps$, then it outputs \accept with probability at least \verynew{$8/15$};
    \item if $\distiso{f}{g} > \eps$, then it outputs \reject with probability at least $8/15$.
  \end{itemize}
  Moreover, one can take \new{$c=\frac{1}{1750}$}, and \new{$\eps_0\eqdef \frac{16}{15} (5-2\sqrt{6}) \simeq 0.108$}. \end{proposition}

\noindent\autoref{theo:iso:testing:robust:kclose:detailed} follows by the combination of~\autoref{lemma:testing:iso:linear:search} and~\autoref{prop:iso:testing:robust:junta}.\medskip

\BPFOF{\autoref{theo:iso:testing:robust:kclose:detailed}}
Let $\rho \eqdef 1-\frac{1}{\sqrt{2}}$, and $\eps'=c \eps$. The algorithm proceeds as follows: it first invokes~\autoref{alg:testing:iso:linear:search} with inputs $f, g, \eps', \delta/2$, and gets by \autoref{lemma:testing:iso:linear:search}, a value $1\leq \kclose \leq n$ such that $\kclose(f,g,\eps')\leq \kclose \leq \kclose(f,g,\frac{\rho\eps'}{\new{16}})$ with probability at least $1-\frac{\delta}{2}$. In particular, conditioning on this we are guaranteed that either $f$ or $g$ is $\eps'$-close to some $\kclose$-junta (i.e., by our choice of $c$, one of the functions is $c \eps$-close to $\junta[\kclose]$). It then calls~\autoref{alg:testing:iso:robust:kjunta} with inputs $f,g,\kclose, \eps$ independently $O(\log\frac{1}{\delta})$ times (for probability amplification from $8/15$ to $1-\frac{\delta}{2}$), and accepts if and only if the majority of these executions returned \accept. The correctness of the algorithm follows from \autoref{prop:iso:testing:robust:junta} and the bound on the query complexity follows from the bounds in~\autoref{lemma:testing:iso:linear:search} and~\autoref{prop:iso:testing:robust:junta}.
\EPFOF

\subsubsection{Linear search: finding $\kclose$.}

Let $\Tester$ denote the algorithm of~\autoref{theo:tol:testing:juntas:tradeoff}, with probability of success amplified by standard techniques to $1-\delta$ for any $\delta \in(0,1]$ (at the price of a factor $\bigO{\log\frac{1}{\delta}}$ in its query complexity); and write $q_{\scalebox{.5}{\Tester}}(k,\eps,\rho,\delta) = \bigO{\frac{k\log k}{\eps\rho(1-\rho)^k}\log\frac{1}{\delta}}$ for its query complexity. \autoref{alg:testing:iso:linear:search}, given next, performs the linear search for $\kclose$: we then analyze its correctness and query complexity.

\begin{algorithm}[H]
  \caption{Junta Degree Finder($\oraclequery{f}, \oraclequery{g}, \eps', \rho,\delta$)\label{alg:testing:iso:linear:search}}
  \begin{algorithmic}[1]
      \State Set $\rho\gets 1-\frac{1}{\sqrt{2}}$ and let $\Tester$ be the algorithm of~\autoref{theo:tol:testing:juntas:tradeoff}.
      \For{$k=0$ to $n$}
        \State Call \Tester on $f$ with parameters $k$, $\eps'$, $\rho$, and $3\delta/(2\pi^2(k+1)^2)$.
        \State Call \Tester on $g$ with parameters $k$, $\eps'$, $\rho$, and $3\delta/(2\pi^2(k+1)^2)$.
        \If{either call to \Tester returned \accept} \Return $k$. \EndIf
      \EndFor
      \State \Return $n$
  \end{algorithmic}
\end{algorithm}

\BPFOF{\autoref{lemma:testing:iso:linear:search}}
By a union bound, all executions of $\Tester$ will be correct with probability at least $1-2\sum_{j=1}^\infty \frac{3\delta}{2\pi^2j^2} = 1-\frac{\delta}{2}$. Conditioning on this, the tester will accept for some $k$ between $\kclose(f,g,\eps')$ and \new{$\kclose(f,g,\rho \eps'/16)$}. This is true since as long as we invoke $\Tester$ with values $k$ such that $f$ and $g$ are $\eps'$-far from $\junta[k]$, both invocations of $\Tester$ will reject. Therefore, once we accept, we have that either $f$ or $g$ is at least $\eps'$-close to $\junta$. Hence, $k \geq \kclose(f,g,\eps')$. Also, $\Tester$ is guaranteed to accept on some $k'$ whenever invoked on a function that is \new{$\rho \eps'/16$}-close to $\junta[k']$. By definition, \new{$\kclose(f,g,\rho \eps'/16)$} is such a $k^\prime$ for either $f$ or $g$; hence, \new{$k \leq \kclose(f,g,\rho \eps'/16)$}.
	
In the case that all the executions of $\Tester$ returned correctly, the query complexity is
\[
    q(\eps,f,g) = \sum_{k=0}^{ \new{\kclose(f,g, \frac{\rho \eps}{16})} } 2q_{\scalebox{.5}{\Tester}}\mleft( k,\eps', \rho, \frac{3\delta}{2\pi^2(k+1)^2} \mright).
\]
By the expression of $q_{\scalebox{.5}{\Tester}}$, we get that $q(\eps,f,g)$ is upper bounded by
\[
q(\eps,f,g) \leq \frac{O(1)}{\eps \rho} \sum_{k=1}^{ k^\star } \frac{k\log k \log\frac{k}{\delta}}{(1-\rho)^k} \leq \frac{O(1)}{\eps } (k^\star \log k^\star)^2  2^{k^\star\log\frac{1}{1-\rho}} \log\frac{1}{\delta}
\]
where \new{$k^\star \eqdef \kclose(f,g, \frac{\rho\eps'}{16})$}. In particular, from the choice of $\rho$, we get
$
q(\eps,f,g) \leq  \tcolor{\bigO{2^{\frac{k^\star}{2} + o(k^\star)}  \frac{1}{\eps} \log\frac{1}{\delta} }}.
$
\medskip

\noindent (If not all the executions of the tester are successful, in the worst case the algorithm considers all possible values of $k$, before finally returning $n$. In this case, the query complexity is similarly bounded by \tcolor{$\bigO{ 2^{\frac{n}{2} + o(n)}  \frac{1}{\eps} \log\frac{1}{\delta} }$}.)

\EPFOF

\subsubsection{Noisy samplers and core juntas.}\label{subsec:noisysamp-core}

For a Boolean function $f\colon\bool^n\to\bool$ we denote by $\closejunta[k]{f}:\bool^n \to \bool$ the $k$-junta closest to $f$. That is, the function $h\in\junta[k]$ such that $\dist{f}{h}=\dist{f}{\junta[k]}$ (if this function is not unique, then we define $\closejunta[k]{f}$ to be the first according to lexicographic order). Moreover, following Chakraborty et al.~\cite{CGM:11},  for a $k$-junta $f\in\junta[k]$ (where we assume without loss of generality that $f$ depends on exactly $k$ variables) we define the \emph{core} of $f$, as follows. The core of $f$, denoted $\core{f}\colon\bool^k\to\bool$, is the  restriction of $f$ to its relevant variables (where these variables are numbered according to the natural order); so that for some $i_1\leq \dots \leq i_k\in[n]$ we have
\[
f(x) = \core{f}(x_{i_1},\dots,x_{i_k})
\]
for every $x\in\bool^n$.

\begin{definition}[{\cite[Definition 1]{CGM:11}}]
Let $g\colon \bool^k \to \bool$ be a function and let $\eta,\mu\in[0,1)$. An \emph{$(\eta,\mu)$-noisy sampler for $g$} is a probabilistic algorithm $\tilde{g}$ that on each execution outputs \verynew{a pair} $(x,a)\in\bool^k\times \bool$ such that
\begin{enumerate}[(i)]
\item For all $y\in \bool^k$, $\probaOf{ x=y } \in \left[\frac{1-\mu}{2^k},\frac{1+\mu}{2^k}\right]$;
\item $\probaOf{ a=g(x) } \geq 1-\eta$;
\item the pairs output on different executions are mutually independent. \end{enumerate}
An \emph{$\eta$-noisy sampler} is an $(\eta,0)$-noisy sampler, i.e., one that on each execution selects a uniformly random $x \in \bool^k$.
\end{definition}

Chakraborty et al.~\cite{CGM:11} show how to build an efficient $\bigO{\eps}$-noisy sampler for $\core{\closejunta[k]{f}}$, which is guaranteed to apply as long as $\dist{f}{\junta[k]} = \bigO{{\eps^6}/{k^{10}}}$. In more detail, they first run a modified version of the junta tester from~\cite{blais2009testing}, which, whenever it accepts, also returns some preprocessing information that enables one to build such a noisy sampler. Moreover, they show that this tester will indeed accept any function that is $\bigO{{\eps^6}/{k^{10}}}$-close to $\junta[k]$ (in addition to rejecting those $\eps$-far from it), giving the above guarantee. Using instead (a small modification of) our tolerant tester from~\autoref{sec:tradeoff}, we are able to extend their techniques to obtain the following~--~less efficient, but more robust~--~noisy sampler.

\begin{restatable}[Noisy sampler for close-to-junta functions]{proposition}{propnoisysampler}\label{prop:noisy:sampler}
There are algorithms $\Algo_P,\Algo_S$ (respectively preprocessor and sampler), which both require oracle access to a function $f\colon\bool^n\to\bool$, and satisfy the following properties.

\begin{itemize}
  \item The preprocessor $\Algo_P$ takes $\eps'\in(0,1]$, \new{$\rho\in(0,1)$}, $k\in\N$ as inputs, makes $\bigO{\frac{k\log\frac{k}{\eps'}}{\eps'\rho(1-\rho)^k}}$ queries to $f$, and either returns \fail or a state $\sigma\in\{0,1\}^{\poly(n)}$. The sampler $\Algo_S$ takes as input such a state $\sigma\in\{0,1\}^{\poly(n)}$, makes a single query to $f$, and outputs a pair $(x,a)\in\bool^k\times\bool$.
  We say that a state $\sigma$ is \emph{$\gamma$-good} if for some permutation $\pi\in\mathcal{S}_k$, $\Algo_S(\sigma)$ is a $\gamma$-noisy sampler for $\core{\closejunta[k]{f}}\circ\pi$.
  \item $\Algo_P(\eps',\rho,k)$ fulfills the following conditions:
    \begin{enumerate}[(i)]
      \item\label{noisy:sampler:completeness} If $\dist{f}{\junta[k]} \leq \frac{\rho}{\new{16}}\eps'$,  then with probability at least $4/5$, $\Algo_P$ returns a state $\sigma$  that is \new{$3\eps'$}-good.
      \item If $\dist{f}{\junta[k]} > \eps'$,  then with probability at least $4/5$, $\Algo_P$ returns \fail.
      \item If  $\dist{f}{\junta[k]} \leq \eps'$, then with probability at least $4/5$, $\Algo_p$ either returns \fail or returns  a state $\sigma$  that is \new{$3\eps'$}-good.
    \end{enumerate}
\end{itemize}
\end{restatable}
The proof of \autoref{prop:noisy:sampler} is deferred to~\autoref{app:noisy:sampler}; indeed, it is almost identical to the proof of Proposition 4.16 in \cite{CGM:11}, with small adaptations required to comply with the use of the tolerant tester from~\autoref{sec:tradeoff} instead of the tester from~\cite{blais2009testing}.

We note that the main difference between the guarantees of our noisy sampler and those of the noisy sampler in~\cite[Lemma 2]{CGM:11} lies in the set of functions for which the noisy sampler is required to return a good state. In our case, this set consists of functions that are 	 \emph{somewhat} close to $k$-juntas. In comparison, the construction from~\cite{CGM:11} is more query-efficient (only $\tildeO{{k}/{\eps}}$ queries to $f$ in the preprocessing stage), but only guarantees the output of a noisy sampler for functions $f$ that are $\bigO{{\eps^6}/{k^{10}}}$-close to $\junta[k]$.\medskip

With these primitives in hand, we are almost ready to prove the main proposition of this subsection, \autoref{prop:iso:testing:robust:junta}. To state the algorithm (\autoref{alg:testing:iso:robust:kjunta}) and proceed with its analysis, we will require the following definition:
\begin{definition}[Number of violating pairs $V_\pi$]
Given two sets $Q_1,Q_2\subseteq \bool^k\times \bool$ and a permutation $\pi\in \mathcal{S}_k$ we say that pairs $(x,a_1)\in Q_1$ and $(y,a_2)\in Q_2$ are \emph{violating with respect to $\pi$}, if $y=\pi(x)$ and $a_1\neq a_2$. We denote the number of violating pairs with respect to $\pi$ by $V_\pi$.
\end{definition}

\begin{algorithm}[H]
\caption{Tolerant isomorphism testing to an unknown $f$ such that $\dist{f}{\junta[k]}\leq c\eps$ ($\oraclequery{f}, \oraclequery{g}, \eps, k$)\label{alg:testing:iso:robust:kjunta}}
  \begin{algorithmic}[1]
  \State Let $\Algo_P,\Algo_S$ be as in~\autoref{prop:noisy:sampler}, $\rho \gets 1- \frac{1}{\sqrt{2}}$, $\eps'\gets\frac{\eps}{\new{16}}$, $\alpha \gets 4c\eps$.
  \State $s\gets C\frac{2^{k/2}}{\eps}\sqrt{k\ln k}$, $t\gets (3\alpha+9\eps')\frac{s^2}{2^k}$. \Comment{\;\;$C>1$ is an absolute constant.}
  \State Run the preprocessor $\Algo_P$ on $f$ and $g$ with parameters $\eps'$, $\rho$, $k$.  \label{step:robust:iso:juntatest}
  \If{either invocation of $\Algo_P$ returned \fail}
	  \State\label{step:robust:iso:juntatest:reject} \Return\reject.
  \EndIf
  \State Using the \new{$3\eps'$}-noisy sampler $\Algo_S$ (called with the states returned on Step~\ref{step:robust:iso:juntatest}), construct ``core'' sets $Q_f,Q_g\subseteq \bool^k\times \bool$ each of size $s\gets C\frac{2^{k/2}}{\eps}\sqrt{k\ln k}$. \label{step:robust:iso:getcoresets}
  \If {there exist $\pi\in \mathcal{S}_k$ such that $V_\pi\leq t$ }
      \State\Return\accept.
  \EndIf
  \State \Return\reject.

  \end{algorithmic}
\end{algorithm}

\BPFOF{\autoref{prop:iso:testing:robust:junta}}
	The query complexity is the sum of the query complexities from Steps~\ref{step:robust:iso:juntatest} and~\ref{step:robust:iso:getcoresets}, i.e.,
	\[
	\bigO{\frac{k\log\frac{k}{\eps}}{\eps\rho(1-\rho)^k}} + 2s\cdot 1 = \bigO{\frac{2^{k/2}}{\eps}k\log\frac{k}{\eps} + \frac{2^{k/2}}{\eps}\sqrt{k\ln k}}
	= \bigO{\frac{2^{k/2}}{\eps}k\log\frac{k}{\eps}}.
	\]
	\paragraph{Completeness.}
	Assume that $g$ is $c\eps$-close to isomorphic to $f$, which itself is $c\eps$-close to being a $k$-junta. Therefore, by the triangle inequality and by our choice of \new{$c \leq \frac{\rho}{512}$}, $\dist{g}{\junta[k]}\leq 2c \eps  \leq \rho\eps'/\new{16}$ as well, so that with probability at least $3/5$ the algorithm does not output $\reject$ on Step~\ref{step:robust:iso:juntatest:reject} (we thereafter analyze this case). Moreover, by the triangle inequality there exists a permutation $\pi\in\mathcal{S}_n$ such that $\dist{\closejunta{f}}{\closejunta{g}\circ\pi} \leq 2c\eps + 2c \eps =4c\eps \eqdef \alpha$. In particular, this implies that there exists a permutation $\pi^\ast\in\mathcal{S}_k$ such that $\dist{\core{\closejunta{f}}}{\core{\closejunta{g}}\circ\pi^\ast} \leq \alpha$. Let $T^\ast \subseteq \bool^k$ be the disagreement set between $\core{\closejunta{f}}$ and $\core{\closejunta{g}}\circ\pi^\ast$: by the above $\abs{T^\ast} \leq \alpha 2^k$.

	Let $Q_f^s, Q_g^s \subseteq \bool^k$ denote the sets resulting from taking the first element in each pair in $Q_f$ and $Q_g$ respectively. The size of the intersection $Z\eqdef\abs{Q_f^s\cap T^\ast}$ is distributed as a \newer{Binomial} random variable, namely
	$Z \sim \newer{\binomial{s}{\frac{\abs{T^\ast}}{2^k}}}$, and conditioned on $Z$ we have $Z^\ast\eqdef\abs{Q_f^s\cap Q_g^s \cap T^\ast}\sim \newer{\binomial{s}{\frac{Z}{2^k}}}$. In particular, we get
	\[
	\expect{Z} = \frac{s \abs{T^\ast}}{2^k}, \quad
	\expectCond{Z^\ast  }{ Z} =  \expectCond{\abs{Q_f^s\cap Q_g^s\cap T^\ast}}{ \abs{Q_f^s\cap T^\ast} } = \frac{s Z }{2^{k}} \;.
		\]
						
	Let $\Algo_S^f$ denote the noisy sampler algorithm when invoked for $f$, and for every $x \in Q_f^s $ let $\Algo_S^f(x)$ denote the label given to $x$ by $\Algo_S^f$. Since $\Algo_S^f$ is a $\new{3\eps'}$-noisy sampler for $\core{\closejunta{f}}$, $\Pr[\Algo_S^f(x) \neq \core{\closejunta{f}}(x)] \leq \new{3\eps'}$. An analogous statement holds for $g$. We let
	$N\eqdef\dabs{ \{x \in Q_f^s \cap Q_g^s \;\colon\; \Algo_S^f(x) \neq \core{\closejunta{f}}(x) \text{ or } \Algo_S^g(x) \neq \core{\closejunta{g}}(x) \} }$ be the number of common samples incorrectly labelled by either noisy sampler,	and observe that  $N$ is dominated by a Binomial random variable $\tilde{N}\sim \binomial{ \abs{Q_f^s\cap Q_g^s} }{ \new{6\eps'} }$.
	
	\noindent With this in hand, we can bound $\probaOf{ V_{\pi^\ast} > t }$ as follows (recall that $t=3\alpha+9\eps'$):
	\begin{align*}
	\probaOf{ V_{\pi^\ast} > (3\alpha+9\eps')\frac{s^2}{2^k} }
			&\leq   \probaOf{ \abs{Q_f^s\cap Q_g^s\cap T^\ast} >  3\alpha\frac{s^2}{2^k} } + \probaOf{ N > 9\eps' \frac{s^2}{2^k} } \\
	&\leq   \probaOf{ \abs{Q_f^s\cap Q_g^s\cap T^\ast} >  3\alpha\frac{s^2}{2^k} } + \probaOf{ \tilde{N} > 9\eps' \frac{s^2}{2^k} } .
	\end{align*}
	Recall that $Z^\ast = \abs{Q_f^s\cap Q_g^s\cap T^\ast} $. Since $\probaOf{ \abs{Q_f^s\cap Q_g^s\cap T^\ast} >  3\alpha\frac{s^2}{2^k} }$ is maximized when $\abs{T^\ast}$ is maximal, we assume without loss of generality that $\abs{T^\ast} = \alpha 2^k$.
	We will handle each term separately. 	
	\begin{align*}
	\probaOf{ Z^\ast >  \frac{3}{2} \cdot \alpha\frac{s^2}{2^k} }
	&= \probaOf{ Z^\ast >  \frac{3}{2}\frac{s^2\abs{T^\ast}}{2^{2k}} }
	\\&= \probaCond{ Z^\ast >  \frac{3}{2}\frac{s^2\abs{T^\ast}}{2^{2k}} }{ Z > \frac{5}{4} \frac{s\abs{T^\ast}}{2^k} } \cdot \probaOf{ Z > \frac{5}{4} \frac{s\abs{T^\ast}}{2^k} }
	\\ &\phantom{=}+
	\probaCond{ Z^\ast >  \frac{3}{2}\frac{s^2\abs{T^\ast}}{2^{2k}} }{ Z \leq \frac{5}{4} \frac{s\abs{T^\ast}}{2^k} } \cdot \probaOf{ Z \leq \frac{5}{4} \frac{s\abs{T^\ast}}{2^k} }
	\\& \leq \probaOf{ Z > \frac{5}{4} \frac{s\abs{T^\ast}}{2^k} }  + \probaCond{ Z^\ast >  \frac{3}{2}\frac{s^2\abs{T^\ast}}{2^{2k}} }{ Z \leq \frac{5}{4} \frac{s\abs{T^\ast}}{2^k} } \;.
	\end{align*}
	We again bound the two terms separately. By the assumption that $\abs{T^\ast} = \alpha2^k$ and by the choice of~$s$, 	\begin{align*}
		\probaOf{ Z > \frac{5}{4} \frac{s\abs{T^\ast}}{2^k} }  = \probaOf{ Z > \frac{5}{4} \expect{Z} } < \exp\left(-\frac{1}{2}\cdot \left(\frac{1}{4}\right)^2 \cdot \frac{s\abs{T^\ast}}{2^k}\right) < \frac{1}{30} .
	\end{align*}
	As for the second term, since $\expect{Z^\ast}= \frac{s Z}{2^k}$ and by the assumption on $T^\ast$ and the setting of $s$,
	\begin{align*}
	 \probaCond{ Z^\ast >  \frac{3}{2}\frac{s^2\abs{T^\ast}}{2^{2k}} }{ Z < \frac{5}{4} \frac{s\abs{T^\ast}}{2^k} }
	 &\leq 	\probaCond{ Z^\ast >  \frac{3}{2}\frac{s^2\abs{T^\ast}}{2^{2k}} }{ Z = \frac{5}{4} \frac{s\abs{T^\ast}}{2^k} }
	 \\ &= \probaCond{ Z^\ast >  \frac{6}{5}\expect{Z^\ast} }{ Z = \frac{5}{4} \frac{s\abs{T^\ast}}{2^k} }
	 \\ &< \exp\left(-\frac{1}{2}\cdot \left(\frac{1}{5}\right)^2 \cdot \frac{s}{2^k} \frac{s\abs{T^\ast}}{2^k}\right) < \frac{1}{30}
	 \end{align*}
	 for a sufficiently large constant $C$ in the definition of $s$.

\noindent As for \new{the last term of the initial expression,} since $\expect{\tilde{N}} = \new{6} \eps' \abs{Q_f^s \cap Q_g^s}$ we have,
	\begin{align*}
	\probaOf{ \tilde{N} > 9\eps' \frac{s^2}{2^k} }
	&\leq \probaCond{ \tilde{N} > 9\eps' \frac{s^2}{2^k} }{ \abs{Q_f^s\cap Q_g^s} \leq \frac{5}{4}\frac{s^2}{2^k} }\cdot\probaOf{ \abs{Q_f^s\cap Q_g^s} \leq \frac{5}{4}\frac{s^2}{2^k} }
	+ \probaOf{ \abs{Q_f^s\cap Q_g^s} > \frac{5}{4}\frac{s^2}{2^k} } \\
	&\leq \probaCond{ \tilde{N} >  9\eps' \frac{s^2}{2^k}}{ \abs{Q_f^s\cap Q_g^s} = \frac{5}{4}\frac{s^2}{2^k} } + \probaOf{ \abs{Q_f^s\cap Q_g^s} > \frac{5}{4}\frac{s^2}{2^k} } \\
	&\leq \probaCond{ \tilde{N} > \frac{6}{5}\expect{\tilde{N}} }{ \abs{Q_f^s\cap Q_g^s} = \frac{5}{4}\frac{s^2}{2^k} } + \probaOf{ \abs{Q_f^s\cap Q_g^s} > \frac{5}{4}\frac{s^2}{2^k} } \\
	&< \exp\left(-\frac{1}{2} \cdot \left(\frac{1}{6}\right)^2 \cdot  \frac{16 \eps' \cdot 5s^2 }{4 \cdot 2^k} \right)
	+ \exp\left(-\frac{1}{2} \cdot \left(\frac{1}{4}\right)^2 \cdot  \frac{s^2 }{2^k} \right) \leq \frac{1}{15} \tag{Actually $o(1)$.}.
	\end{align*}
	The algorithm will therefore reject with probability at most $\frac{2}{5}+\frac{1}{15}+\frac{1}{15} = \frac{7}{15}$.

	\paragraph{Soundness.} Assume that $\dist{f}{\junta} \leq \verynew{c\eps}$, and that $g$ is $\eps$-far from being isomorphic to $f$. Then one of the following must hold:
	\begin{enumerate}
		\item $\dist{g}{\junta} > \eps'$.
		\item for all $\pi\in \mathcal{S}_k$, $\dist{\core{\closejunta{f}}}{\core{\closejunta{g}}\circ{\pi}} > \eps - (\eps' + c \eps) > \eps-2\eps'$.
	\end{enumerate}
	If the first case holds, then the function will be rejected in Step~\ref{step:robust:iso:juntatest} with probability at least $\frac{4}{5}$, and so the algorithm will reject as desired. We can therefore focus on the second case.
	
	If the second case holds, either the tester rejects in Step~\ref{step:robust:iso:juntatest:reject} (and we are done) or it outputs a state which will be used to get the $\new{3\eps'}$-noisy sampler. Fix any $\pi\in \mathcal{S}_k$. Since $\dist{\core{\closejunta{f}}}{\core{\closejunta{g}}\circ{\pi}} > (\eps-2\eps')$, there are $m\eqdef m(\pi) \geq (\eps-2\eps')2^k$ inputs $x\in\bool^k$ such that $\core{\closejunta{f}}(x)\neq\core{\closejunta{g}}\circ{\pi}(x)$. Let $T=T(\pi)\subseteq \bool^k$ denote the set of all such inputs (so that $\abs{T}=m$).
	
	We can make a similar argument as for the completeness case: we have that $\abs{Q_f^s\cap T}$ is a random variable with \newer{Binomial} distribution (of parameters $s$, \newer{and $\frac{\abs{T}}{2^k}$}).  	Conditioned on $\abs{Q_f^s\cap T}$, we also have $\abs{Q_f^s\cap Q_g^s\cap T}\sim\newer{\binomial{s}{\frac{\abs{Q_f^s\cap T}}{2^k}}}$, so that
	\[
	\expect{\abs{Q_f^s\cap Q_g^s\cap T}} = \expect{ \expectCond{ \abs{Q_f^s\cap Q_g^s\cap T} }{ \abs{Q_f^s\cap T} } }
	= \expect{ \frac{s \abs{Q_f^s\cap T}}{2^k} }
	= \frac{s^2 \abs{T}}{2^{2k}}  \geq (\eps-2\eps')\frac{s^2}{2^k} = \new{14}\eps'\frac{s^2}{2^k}.
	\]
	(Recall that our threshold was set to $t=(3\alpha+9\eps')\frac{s^2}{2^k} \leq \new{12}\eps'\frac{s^2}{2^k}$.) Moreover, each element  $x \in Q_f^s\cap Q_g^s\cap T$ will contribute to $V_\pi$ with probability at least $\left(1-\new{3\eps'}\right)^2 > \frac{24}{25}$ (since this is a lower bound on the probability that both $\Algo_S^f(x)= \core{\closejunta{f}} (x)$ and $\Algo_S^g(x)= \core{\closejunta{g}} (x)$, \new{and as $\eps' \leq \frac{\eps_0}{16}$}). As before, we can therefore write, letting $Z\eqdef \abs{Q_f^s\cap Q_g^s\cap T}$, and taking $\abs{T}$  to be minimal  so that $\abs{T}=(\eps - 2\eps')2^k$,
	\begin{align*}
	\probaOf{ V_{\pi} > t }
	&\geq \probaCond{ V_{\pi} > t }{ Z \geq \frac{13}{12}t }\probaOf{ Z \geq \frac{13}{12}t }\\
	&\geq (1-e^{-\frac{1}{2}\left(\frac{1}{26}\right)^2\cdot\frac{24}{25}\cdot\frac{13t}{12}})\probaOf{ Z \geq \frac{13}{12}t } \tag{Chernoff bound} 	= \left(1-e^{-\frac{t}{1300}}\right)\probaOf{ Z \geq \frac{13}{12}t }
	\end{align*}
	so that it is sufficient to lower bound $\probaOf{ Z \geq \frac{13}{12}t }$. To do so, we will bound the probability of the two following events:
	\begin{description}
		\item[E1:] $Y\eqdef \abs{Q_f^s\cap T} < \frac{99}{100} \frac{s\abs{T}}{2^k}$
		\item[E2:] $Z=\abs{Q_f^s\cap Q_g^s\cap T} < \frac{99}{100} \frac{s}{2^k}\abs{Q_f^s\cap T}$, conditioning on $\abs{Q_f^s\cap T} \geq \frac{99}{100} \frac{s\abs{T}}{2^k}$.
	\end{description}
	This will be sufficient for us to conclude, as by our choice of $t=(3\alpha+9\eps')\frac{s^2}{2^k}$,the setting $|T|=(\eps -2\eps')2^k=\frac{7}{8}\eps2^k$, and since $\alpha \leq \eps'$,
	we have 
\[\frac{13}{12}\cdot t=\frac{13}{12}\cdot 12\eps'\cdot \frac{s^2}{2^k}=\frac{13}{16}\cdot \eps \cdot \frac{s^2}{2^k} \leq \left(\frac{99}{100}\right)^2\frac{s^2|T|}{2^{2k}}\;.\]	
Therefore, by a Chernoff bound
	\begin{align*}
	\probaOf{ Z < \frac{13}{12}t }
	&\leq \probaOf{ Z < \left(\frac{99}{100}\right)^2 \frac{\new{s^2\abs{T}}}{2^{2k}} } \\
	&\leq  \probaOf{ Y < \frac{99}{100} \frac{s\abs{T}}{2^k} } +
	\probaCond{ Z < \frac{99}{100} \frac{s}{2^k} Y  }{ Y \geq \frac{99}{100} \frac{s\abs{T}}{2^k} }\probaOf{ Y \geq \frac{99}{100} \frac{s\abs{T}}{2^k} }
	\\ &\leq  \probaOf{ Y < \frac{99}{100} \frac{s\abs{T}}{2^k} } +
	\probaCond{ Z < \frac{99}{100} \frac{s}{2^k} Y  }{ Y \geq \frac{99}{100} \frac{s\abs{T}}{2^k} }
	\\ &< \exp\left(- \frac{1}{2} \cdot \left(\frac{1}{100}\right)^2 \cdot \frac{s\abs{T}}{2^k} \right) +
	 \probaCond{ Z < \frac{99}{100} \frac{sY}{2^k} }{ Y= \frac{99}{100} \frac{s\abs{T}}{2^k} }
	 \\ &\leq \exp\left(- \frac{1}{2} \cdot \left(\frac{1}{100}\right)^2 \cdot \frac{s (\eps -2\eps')\cdot 2^k}{2^k} \right) +
	 \exp\left(- \frac{1}{2} \cdot \left(\frac{2}{100}\right)^2 \cdot \frac{s^2 (\eps -2\eps')\cdot 2^k}{2^{2k}} \right)
	 \\& < \exp(-\tau C^2 k \ln k),
	 \end{align*}
	by the choice $s=C \frac{2^{k/2}}{\eps}\sqrt{k \ln k}$, and for some constant $\tau\in(0,1)$. Hence setting $C$ to a sufficiently large constant, the foregoing analysis implies that $\probaOf{V_{\pi} \leq t }\leq e^{-\frac{t}{1300}} + e^{-\tau C^2 k \ln k} = e^{-\frac{12c+3/4}{1300}\eps {s^2}/{2^k}} + e^{-\tau C^2 k \ln k} \leq \frac{7}{15k^k}$. A union bound over all $k!< k^k$ permutations $\pi\in\mathcal{S}_k$ finally yields $\probaOf{ \exists \pi,\, V_{\pi} \leq t } \leq \frac{7}{15}$ as claimed.

\EPFOF

	 \section{Acknowledgments}
	The authors would like to thank Lap Chi Lau for useful discussions, and in particular his suggestion to look at~\cite{LSW:15}. Talya Eden is grateful to the Azrieli Foundation for the award of an Azrieli Fellowship.												
	\bibliographystyle{alpha}
	\bibliography{references}
	
	\clearpage
	\appendix
	\makeatletter{}\section{Proof of~\autoref{prop:noisy:sampler} (construction of a noisy sampler)} \label{app:noisy:sampler}

We provide in this appendix the proof of~\autoref{prop:noisy:sampler}, restated below:

\propnoisysampler*

We will very closely follow the argument from the full version of~\cite{CGM:11} (Proposition 4.16),\footnote{The full version can be found at~\url{http://www.cs.technion.ac.il/~ariem/eseja.pdf}.} adapting the corresponding parts in order to obtain our result. For completeness, we tried to make this appendix below self-contained, reproducing almost verbatim several parts of the proof from~\cite{CGM:11}.\footnote{The reader may notice that Chakraborty et al. rely on a definition of set-influence that differs from ours by a factor $2$; we propagated the changes through the argument.}\medskip

\BPFOF{\autoref{prop:noisy:sampler}}
In order to use our result from~\autoref{sec:tradeoff} \textit{in lieu} of the junta tester from~\cite{blais2009testing}, we first need to make a small modification to our algorithm. Specifically, in its first step our tester will now pick a random partition $\mI$ of $[n]$ in $\ell\eqdef \frac{Ck^2}{\eps}$ parts instead of $24k^2$ (for some (small) absolute constant $C>1$). It is easy to check that both~\autoref{lemma:random:partition:soundness} and~\autoref{lemma:any:partition:completeness} still hold (e.g., from the proof of~\cite[Lemma 5.4]{Blais:PhD:12}), now with probability at least $19/20$. Moreover, our modified tolerant tester offers the same soundness and completeness guarantees as~\autoref{theo:tol:testing:juntas:tradeoff}, at the price of a query complexity $\bigO{\frac{k\log({k}/{\eps})}{\eps\rho(1-\rho)^k}}$ (instead of $\bigO{\frac{k\log k}{\eps\rho(1-\rho)^k}}$). Moreover, in Step~\ref{algo:tradeoff:noisy:influence:accept} of~\autoref{algo:tradeoff:noisy:influence}, i.e. when the algorithm found a suitable set $J\subseteq [\ell]$ \new{(of size $\ell-k$)} as a witness for accepting, we make the algorithm return $\mI$ and the set $\mathcal{J}\eqdef\{I_j\}_{j\in \bar{J}}$ along with the verdict $\accept$.\medskip

We will also require the definitions of the distribution induced by a partition $\mI$ and a subset $\mathcal{J}\subseteq \mI$, and of such a couple $(\mI,\mathcal{J})$ being \emph{good} for a function:
\begin{definition}[{\cite[Definition 4.6]{CGM:11}}]
  For any partition $\mI=\{ I_1,\dots,I_\ell \}$ of $[n]$, and subset of parts $\mathcal{J}\subseteq \mI$, we define a pair of distributions:
  \begin{description}
    \item[The distribution $\D_{\mI}$ on $\bool^n$.] An element $y\sim \D_{\mI}$ is sampled by
        \begin{enumerate}
          \item picking $z\in\bool^\ell$ uniformly at random among all $\binom{\ell}{\ell/2}$ strings of weight $\frac{\ell}{2}$;
          \item setting $y_i = z_j$ for all $j\in[\ell]$ and $i\in I_j$.
        \end{enumerate}
      \item[The distribution $\D_{\mathcal{J}}$ on $\bool^{\abs{\mathcal{J}}}$.] An element $x\sim \D_{\mathcal{J}}$ is sampled by
        \begin{enumerate}
          \item picking $y\sim \D_{\mI}$;
          \item outputting $\mathsf{extract}_{(\mI,\mathcal{J})}(y)$, where $x=\mathsf{extract}_{(\mI,\mathcal{J})}(y)$ is defined as follows. For all $j\in[\ell]$ such that $I_j\in\mathcal{J}$:
              \begin{itemize}
                \item if $I_j\neq \emptyset$, set $x_j = y_i$ (where $i\in I_j$);
                \item if $I_j = \emptyset$, set $x_j$ to be a uniformly random bit.
              \end{itemize}
        \end{enumerate}
  \end{description}
\end{definition}

\begin{lemma}[{\cite[Lemma 4.7]{CGM:11}}]\label{lemma:properties:distributions:cgm}
  $\D_{\mI}$ and $\D_{\mathcal{J}}$ as above satisfy the following.
  \begin{itemize}
    \item For all $a\in\bool^n$, $\probaDistrOf{\mI, y\sim \D_{\mI}}{y=a} = \frac{1}{2^n}$.
    \item Assume $\ell > 4\abs{\mathcal{J}}^2$. For every $\mI$ and $\mathcal{J}\subseteq\mI$, the total variation distance between $\D_{\mathcal{J}}$ and the
uniform distribution on $\bool^{\abs{\mathcal{J}}}$ is bounded by $2\abs{\mathcal{J}}^2/\ell$. Moreover, the $\lp[\infty]$ distance between the
two distributions is at most $4\abs{\mathcal{J}}^2/(\ell 2^{\abs{\mathcal{J}}})$.
  \end{itemize}
\end{lemma}

\begin{definition}[{\cite[Definition 4.8]{CGM:11}}]\label{def:sampler:cgm}
  Given $(\mI,\mathcal{J})$ as above and oracle access to $f\colon\bool^n\to\bool$, we define a probabilistic algorithm $\mathsf{sampler}_{(\mI,\mathcal{J})}(f)$ that on each execution produces a pair $\dotprod{x}{a} \in \bool^{\abs{\mathcal{J}}}\times\bool$ as follows: first it picks a random $y\sim\D_{\mI}$, then it queries $f$ on $y$, computes $x=\mathsf{extract}_{(\mI,\mathcal{J})}(y)$ and outputs the pair
$\dotprod{x}{f(y)}$.
\end{definition}

\begin{definition}[{\cite[Definition 4.9]{CGM:11}}]\label{def:noisysampler:good}
  Given $\alpha>0$, a function $f\colon\bool^n\to\bool$, a partition $\mI=\{ I_1,\dots,I_\ell \}$ of $[n]$ and
  a subset $\mathcal{J}\subseteq \mI$ of $k$ parts, we call the pair $(\mI,\mathcal{J})$ \emph{$\alpha$-good} (with respect to $f$) if there exists a $k$-junta $h\in\junta[k]$ such that the following conditions are satisfied:
  \begin{enumerate}
    \item Conditions on $h$:
          \begin{enumerate}
            \item Every relevant variable of $h$ is also a relevant variable of $\closejunta[k]{f}$;
            \item $\dist{h}{\closejunta{f}} \leq \alpha$.
          \end{enumerate}
    \item Conditions on $\mI$:
          \begin{enumerate}
            \item For all $j\in[\ell]$, $I_j$ contains at most one variable of $\core{\closejunta[k]{f}}$;
            \item $\probaDistrOf{y\sim\D_{\mI}}{f(y)\neq \closejunta[k]{f}(y)} \leq 10\cdot\dist{f}{\closejunta[k]{f}}$.
          \end{enumerate}
    \item Condition on $\mathcal{J}$: the set $S\eqdef \bigcup_{I\in \mathcal{J}} I$ contains all relevant variables of $h$.
  \end{enumerate}
\end{definition}

\begin{lemma}[{\cite[Lemma 4.10]{CGM:11}}]\label{lemma:noisysampler:good}
  Let $\alpha,f,\mI,\mathcal{J}$ be as in the preceding definition. If the pair $(\mI,\mathcal{J})$ is $\alpha$-good (with respect
to $f$), then $\mathsf{sampler}_{(\mI,\mathcal{J})}(f)$ (as per~\autoref{def:sampler:cgm}) is an $(\eta,\mu)$-noisy sampler for some permutation of $\core{\closejunta[k]{f}}$,
with $\eta \leq 2\alpha+\frac{4k^2}{\ell}+10\cdot\dist{f}{\closejunta[k]{f}}$ and $\mu\leq\frac{4k^2}{\ell}$.
\end{lemma}

The last piece we shall need is the ability to convert an $(\eta,\mu)$-noisy sampler to a $(\eta',0)$-noisy sampler -- that is, one whose samples are exactly uniformly distributed.
\begin{lemma}[{\cite[Lemma 4.4]{CGM:11}}]\label{lemma:noisysampler:uniform:conversion}
  Let $\tilde{g}$ be an $(\eta,\mu)$-noisy sampler for $g\colon\bool^k\to\bool$, that on each execution picks $x$ according to some fixed (and fully known) distribution $\D$. Then $\tilde{g}$ can be turned into an $(\eta+\mu)$-noisy sampler $\tilde{g}_{\rm unif}$ for $g$.
\end{lemma}

With this in hand, we are ready to prove the main lemma:
\begin{lemma}[{Analogue of~\cite[Proposition 4.16]{CGM:11}}]\label{lemma:analogue:cgm}
The tester from~\autoref{theo:tol:testing:juntas:tradeoff}, modified as above, as the following guarantees. It has query complexity $\bigO{\frac{k\log({k}/{\eps})}{\eps\rho(1-\rho)^k}}$ and outputs, in case of acceptance, a partition $\mI$ of $[n]$ in $\ell\eqdef \bigO{{k^2}/{\eps}}$ parts along with a subset $\mathcal{J}\subseteq \mI$ of $k$ parts such that for any $f$ the following conditions hold:
\begin{itemize}
  \item if $\dist{f}{\junta[k]}\leq \frac{\rho}{\new{16}}\eps$, the algorithm accepts with probability at least $9/10$;
  \item if $\dist{f}{\junta[k]}> \eps$, the algorithm rejects with probability at least $9/10$;
  \item for any $f$, with probability at least $4/5$ either the algorithm rejects, or it outputs $\mathcal{J}$ such that the pair $(\mI,\mathcal{J})$ is \new{$\frac{1}{2}(1+\frac{3}{8}\rho)\eps$}-good (as per~\autoref{def:noisysampler:good}).
\end{itemize}
In particular, if $\dist{f}{\junta[k]}\leq \frac{\rho}{\new{16}}\eps$, then with probability at least $4/5$ the algorithm outputs a set $\mathcal{J}$ such
that $(\mI,\mathcal{J})$ is \new{$\frac{1}{2}(1+\frac{3}{8}\rho)\eps$}-good.
\end{lemma}

\BPFOF{\autoref{lemma:analogue:cgm}}
The first two items follow from the analysis of the tester (\autoref{theo:tol:testing:juntas:tradeoff}) and the foregoing discussion; we thus turn to establishing the third item. 

Called with parameters $k,\rho,\eps$, our algorithm, with probability at least $19/20$, either rejects or outputs a partition $\mI$ of $[n]$ into $\ell=\bigO{k^2}$ parts and set $\mathcal{J}\subseteq\mI$ satisfying \new{$\infl[f]{\overline{\phi(J)}} \leq \eps$}. Let $R\subseteq [n]$ (with $\abs{R}\leq k$) denote the set of relevant variables of $\closejunta[k]{f}$, and $V\supseteq R$ (with $\abs{V}=k$) the set of relevant variables of $\core{\closejunta[k]{f}}$. Assume that $\dist{f}{\junta[k]} \leq \frac{\rho\eps}{\new{16}}$.\footnote{For other $f$'s, the third item follows from the second item.} We then have:
  \begin{itemize}
    \item by the above, with probability at least $19/20$ the algorithm outputs a set $\mathcal{J}\subseteq\mI$      which satisfies
    \[
      \infl[f]{\overline{\phi(J)}} \leq \new{\eps};
    \]
    \item since $\ell \gg k^2$, with probability at least $19/20$ all elements of $V$ fall in  different parts of the partition $\mI$;
    \item by~\autoref{lemma:properties:distributions:cgm} and by Markov's inequality, with probability at least $9/10$ the partition $\mI$ satisfies $\probaDistrOf{y\sim\mathcal{D}_{\mI}}{ f(y) \neq \closejunta[k]{f}(y) } \leq 10\cdot \dist{f}{\closejunta[k]{f}}$.
  \end{itemize}
  So by a union bound, with probability at least $4/5$ all three of these events occur. Now we show that conditioned on them, the pair $(\mI,\mathcal{J})$ is $(1+\frac{3}{2}\rho)\eps$-good.
Let $U\eqdef R\cap \left( \bigcup_{I\in\mathcal{J}} I\right)$ (informally, U is the subset of the relevant variables of $\closejunta[k]{f}$ that were successfully ``discovered'' by the tester). Since $\dist{f}{\junta[k]}\leq \frac{\rho\eps}{\new{16}}$, we have $\infl[f]{\bar{V}} \leq \new{4}\dist{f}{\junta[k]} \leq \frac{\rho\eps}{\new{4}}$. By the subadditivity and monotonicity of influence we get
\[
    \infl[f]{\bar{U}} \leq \infl[f]{\bar{V}} + \infl[f]{V\setminus U} \leq \infl[f]{\bar{V}} + \infl[f]{\overline{\phi(J))}} \leq \frac{\rho\eps}{\new{4}}+\new{\eps}.
\]
where the second inequality follows from $V\setminus U\subseteq \overline{\phi(J)}$. This means (see e.g.~\cite[Lemma 2.21]{Blais:PhD:12}) that there is a $k$-junta $h$ on $U$ such that $\dist{f}{h}\leq \new{\frac{1}{2}(\frac{\rho\eps}{4}+\eps)}$, and by the triangle inequality $\dist{\closejunta[k]{f}}{h} \leq \new{\frac{1}{2}(\frac{\rho\eps}{4}+\eps)+\frac{\rho\eps}{16}} = \new{\frac{1}{2}(1+\frac{3}{8}\rho)\eps}$. Based on this $h$, we can verify that the pair $(\mI,\mathcal{J})$ is  $\new{\frac{1}{2}(1+\frac{3}{8}\rho)\eps}$-good by going over the conditions in~\autoref{def:noisysampler:good}.
\EPFOF
  
\paragraph{Concluding the proof of~\autoref{prop:noisy:sampler}.}
We conclude as in Section~4.6 of~\cite{CGM:11}, and start by describing how $\Algo_P$ and $\Algo_S$ operate. The preprocessor $\Algo_P$ starts by calling the tester $\Tester$ of~\autoref{lemma:analogue:cgm}. Then, in case $\Tester$ accepted, $\Algo_P$ encodes in the state $\sigma$ the partition $\mI$ and the subset $\mathcal{J}\subseteq \mI$ output by $\Tester$ (see~\autoref{lemma:analogue:cgm}), along with the values
of $k$ and $\eps$. The sampler $\Algo_S$, given $\sigma$, obtains a pair $\dotprod{x}{a} \in \bool^k\times\bool$ by executing $\mathsf{sampler}_{(\mI,\mathcal{J})}(f)$ (from~\autoref{def:sampler:cgm}) once.
Now we show how~\autoref{prop:noisy:sampler} follows from~\autoref{lemma:analogue:cgm}. The first two items are immediate. As for the third item, notice that we only have to analyze the case where $\dist{f}{\closejunta[k]{f}}\leq  \frac{\rho\eps}{\new{16}}$ and $\Tester$ accepted; all other cases are taken care of by the first two items. By the third item in~\autoref{lemma:analogue:cgm}, with probability at least $4/5$ the pair $(\mI,\mathcal{J})$ is \new{$\frac{1}{2}(1+\frac{3}{8}\rho)\eps$}-good. If so, by~\autoref{lemma:noisysampler:good}, $\mathsf{sampler}_{(\mI,\mathcal{J})}(f)$ is
an $(\eta,\mu)$-noisy sampler for some permutation of $\core{\closejunta[k]{f}}$, where 
\[
  \eta
  \leq 2\cdot\new{\frac{1}{2}(1+\frac{3}{8}\rho)\eps}+\frac{4k^2}{\ell}+10\cdot\dist{f}{\junta[k]}
  \leq (1+\frac{3}{8}\rho)\eps + \frac{10\rho\eps}{16}+\frac{4k^2}{\ell}
  = \new{(1+\rho)\eps+\frac{4k^2}{\ell}}
\] and $\mu \leq \frac{4k^2}{\ell}$. This in turn implies by~\autoref{lemma:noisysampler:uniform:conversion} an $\eta'$-noisy sampler, for
\[
    \eta' = \eta+\mu \leq \new{(1+\rho)}\eps+\frac{8k^2}{\ell} \leq \new{2+\rho)}\eps \leq \new{3}\eps
\]
as claimed. (Where we used that $\frac{8k^2}{\ell} \leq \eps$ by our choice of $\ell$.)
\EPFOF
 
	\makeatletter{}\section{Proof of~\autoref{claim:legal:collection:bound}} \label{app:baranyai}

We provide in this appendix the proof of~\autoref{claim:legal:collection:bound}, restated below:

\begin{definition}
Let $X$ be a set of $j$ elements, and for any $s \in [j]$ consider the family $\subsetcoll{X}{s}$ of
all subsets of $X$ that have size $s$. We shall say that $\mathcal{C} \subseteq \subsetcoll{X}{s}$ is a \emph{cover} of $X$,
if $\bigcup_{Y\in \mathcal{C}} Y = X$.  We shall say that $\mathcal{C}_1,\dots,\mathcal{C}_m$ is a \emph{legal collection of covers for $X$}, if each $\mathcal{C}_t$ is a cover of $X$, and these covers are disjoint.
\end{definition}

\begin{claim}
For any set $X$ of $j$ elements, there exists a legal collection of covers for $X$ of size at least 
\[
m \geq \flr{\frac{{\binom{j}{s}} }{ \clg{\frac{j}{s}} }}.
\]
(Moreover, this bound is tight.)
\end{claim}

\BPF
This claim follows from a result due to Baranyai~\cite{Baranyai:75} on factorization of regular hypergraphs. We state this result, and describe how to derive the claim from it, below (recall that $K_n^h$ denotes the $h$-regular hypergraph $K_n^h$ on $n$ vertices):\footnote{An exposition of this result and the original proof as given by Baranyai can also be found in~\cite[Theorem 4.1.1]{Brandt:2015:thesis}.}
 
  \begin{theorem}[Baranyai's Theorem {\cite[Theorem 1]{Baranyai:75}}]\label{theo:combinatorics:baranyai}
    Let $n,h$ be integers satisfying $1\leq h\leq n$, and $a_1,\dots,a_\ell$ integers such that $\sum_{i=1}^\ell a_i = \binom{n}{h}$. Then the edges of $K_n^h$ can be partitioned in hypergraphs $\mathcal{H}_1\dots,\mathcal{H}_\ell$ such that 
    \begin{enumerate}[(i)]
      \item $\abs{\mathcal{H}_i} = a_i$ for all $i\in[\ell]$;
      \item each $\mathcal{H}_i$ is \emph{almost regular}: the number of hyperedges any two vertices $u,v\in \mathcal{H}_i$ participe in differ by at most one (and here, specifically, is either $\clg{\frac{a_i h}{n}}$ or $\flr{\frac{a_i h}{n}}$).
    \end{enumerate}
  \end{theorem}
  We apply~\autoref{theo:combinatorics:baranyai} as follows: setting $m\eqdef \flr{\frac{\binom{j}{s}}{\clg{\frac{j}{s}}}} \leq \binom{j-1}{s-1}$ and $\ell\eqdef m+1$, we let $a_i \eqdef \clg{\frac{j}{s}}$ for all $1\leq i\leq m$, and $a_\ell \eqdef \binom{j}{s}- \sum_{i=1}^m a_i\geq 0$. By the theorem, we obtain a partition of $K^j_s$ into $\ell=m+1$ hypergraphs $\mathcal{H}_1\dots,\mathcal{H}_\ell$ such that the first $m$ satisfy:
    \begin{enumerate}[(i)]
      \item $\abs{\mathcal{H}_i} = \clg{\frac{j}{s}}$ for all $i\in[m]$;
      \item for any $i\in[m]$, any vertex $u\in \mathcal{H}_i$ participes in either $1$ or $2$        hyperedges;
    \end{enumerate}
  (and we cannot say much about the ``remainder`` hypergraph $\mathcal{H}_\ell$). Condition (ii) ensures that each of the first $m$ hypergraphs obtained indeed defines a cover of the set of $j$ elements by $s$-element subsets, while by definition of the partition of the hypergraph we are promised that these $m$ covers are disjoint. This proves the lemma, as $\mathcal{H}_1\dots,\mathcal{H}_m$ then induce a legal cover of $X$.
  
  As for the optimality of the bound, it follows readily from observing that one must have 
$
  m \leq \flr{ \frac{\binom{j}{s}}{\clg{ j/s }} }
$
since for every cover $\mathcal{C}$ we must have $\abs{\mathcal{C}} \geq \clg{ j/s }$, and $\abs{\subsetcoll{X}{s}} = \binom{j}{s}$.
\EPF

\end{document}